\documentclass[10pt]{amsart}
\usepackage{pifont}
\usepackage{booktabs, siunitx,  txfonts}

\usepackage{pgf,tikz}
\usepackage{amssymb, latexsym}
\usepackage{graphicx}
\usepackage{pdfsync, wrapfig}
\usepackage{xcolor}
\definecolor{lightgray}{gray}{0.9}
\definecolor{dnrbl}{rgb}{0,0,0.5}
\definecolor{dnrgr}{rgb}{0,0.5,0}
\definecolor{dnrre}{rgb}{0.5,0,0}
\usepackage[colorlinks=true, citecolor=dnrgr, linkcolor=dnrre, urlcolor=dnrbl] {hyperref}
\usepackage{pgf,tikz}
\usetikzlibrary{decorations.pathreplacing,  decorations.shapes,  decorations.pathmorphing}
\usetikzlibrary{decorations}
\usetikzlibrary{arrows,shapes,positioning}
\usetikzlibrary{decorations.markings}
\usepackage{tabularx, colortbl}
\usepackage[scriptsize, up]{caption}

\theoremstyle{plain}
\newtheorem{thm}{Theorem}[section]

\newtheorem{lem}[thm]{Lemma}
\newtheorem{coro}[thm]{Corollary}
\newtheorem{defi}[thm]{Definition}
\numberwithin{equation}{section}

\newcommand{\un}{\uparrow} 
\newcommand{\de}{\downarrow} 

\newcommand{\ml}{Martin-L\"{o}f }

 \newcommand{\ivr}{integer-valued random}
 \newcommand{\pivr}{partial integer-valued random}
 \newcommand{\ivrs}{integer-valued randoms}
 \newcommand{\pivrs}{partial integer-valued randoms}
 \newcommand{\anc}{array noncomputable}

\begin{document}
\title[integer-valued Randomness and Turing Degrees]{Integer 
Valued Betting strategies and Turing Degrees}
\author{George Barmpalias, Rod G.~Downey and Michael McInerney}
\address{{\bf George Barmpalias} (1) State Key Lab of Computer Science, 
Institute of Software,
Chinese Academy of Sciences,
Beijing 100190,
China and (2) School of Mathematics, Statistics and Operations Research,
Victoria University, Wellington, New Zealand}
\email{barmpalias@gmail.com}
\address{{\bf Rod G.~Downey and Michael McInerney:} 
School of Mathematics, Statistics and Operations Research,
Victoria University, Wellington, New Zealand.}
\email{\{rod.downey, michael.mcinerney\}@msor.vuw.ac.nz}

\date{\today}

\thanks{All authors are supported by the Marsden Fund of New Zealand.
Barmpalias was also supported by the 
Research Fund for International Young Scientists 
from the National Natural Science Foundation of China,
grant number 613501-10535.
Research was (partially) completed 
while the authors were 
visiting the Institute for 
Mathematical Sciences, National University of Singapore in July, 2014}

\begin{abstract} 
Betting strategies are often expressed formally as martingales. 
A martingale is called integer-valued if each bet must be an
integer value. 
Integer-valued strategies correspond to the fact
that in most betting situations, there is a minimum amount that
a player can bet. According to a well known paradigm,
algorithmic randomness can be founded on the notion of betting strategies.
A real $X$ is called integer-valued random
if no effective integer-valued martingale
succeeds on $X$.
It turns out that
this notion of randomness has interesting interactions
with genericity and the computably enumerable degrees. 
We investigate the computational power of
the integer-valued random reals in terms of standard notions from
computability theory.
\end{abstract}

\maketitle

\section{Introduction}

An interesting strategy for someone who wishes to make a profit
by betting on the outcomes of a series of unbiased coin tosses, is to
double the the amount he bets each time he places a bet.
Then, independently of whether he bets on heads or tails, if the coin is fair
(i.e.\ the sequence of binary outcomes is random) he is guaranteed to win infinitely many bets.
Furthermore, each time he wins he
recovers all previous losses, plus he wins a profit equal to the original stake.
This is a simple example from a class of betting strategies 
that originated from, and were popular in 18th century France.
 They are known as {\em martingales}. The ``success'' of this strategy is essentially 
 equivalent to the fact that a symmetric one-dimensional random walk will eventually travel an
 arbitrarily long distance to the right of the starting point (as well as an arbitrarily long distance 
 to the left of the starting point).

 So what is the catch? For such a strategy to be maintained, the player needs to be able
 to withstand arbitrarily large losses, and such a requirement is not practically feasible.
 In terms of the random walk, this corresponds to the fact that, before it travels a large distance to the right
 of the starting point, it is likely to have travelled a considerable distance to the left of it.
 
 \subsection{Martingales and randomness}
 Martingales have been reincarnated in probability theory (largely though the work of Doob), 
 as (memoryless) 
 stochastic processes $(Z_n)$ such that the conditional expectation of each $Z_{n+1}$
 given $Z_n$ remains equal to the expectation of $Z_0$.
 The above observations on a fair coin-tossing game 
 are now theorems in the theory of martingales. For example, Doob's maximal 
 martingale inequality
 says that with probability 1, a non-negative martingale is bounded. Intuitively this means that,
 if someone is not able (or willing) to take credit (so that he continues to bet after his balance is negative)
 then the probability that he makes an arbitrarily large amount in profit is 0.
 
Martingales in probability rest on a concept of randomness in order to determine (e.g.\ with high probability)
or explain the outcomes of
stochastic processes. In turns out that this methodology can be turned upside down, so that certain processes
are used in order to define or explain the concept of randomness.
 Such an approach was initiated by
 Schnorr in \cite{Schnorr}, and turned out to be
one of the standard and most intuitive methods of assigning meaning to the concept 
of randomness for an individual string or a real (i.e.\ an 
infinite binary sequence, a 
point in the Cantor space). This approach is often known as
the \emph{unpredictability paradigm}, and
it says that it should not be possible for a computable predictor
to be able to predict bit $n+1$ of a  real
 $X$ based on knowledge of 
bits $1,\dots,n$ of $X$, namely $X\upharpoonright n$. 
The unpredictability paradigm  can be formalized by using \emph{martingales},
which (for our purposes)
can be seen as betting strategies.
We may define a martingale 
 to be a 
function $f:2^{<\omega}\to {\mathbb R}^{\geqslant 0}$  
which obeys 
the following fairness condition:
\[
f(\sigma)=\frac{f(\sigma 0)+f(\sigma 1)}{2}.
\]
If $f$ is partial, but its domain is downward closed with respect to the prefix relation on
finite strings, then we say that $f$ is a {\em partial martingale}.

In probability terms, $f$ can be seen as a 
stochastic process (a series of dependent variables) $Z_s$
where $Z_s$ represents the capital of a player at the end of the $s$th bet
(where there is 50\% chance for head or tails).
Then the fairness condition says that the expectation of $f$ at stage $s+1$ is the same
as the value of $f$ at stage $s$. In other words, the fairness condition says that the expected
growth of $f$ at each stage of this game is 0. 
If we interpret $f$ as the capital of a player
who bets on the outcomes of the coin tosses, the fairness condition says that there is no bias
in this game toward the player or the house.
Moreover note that our definition of a martingale as a betting strategy requires that it is non-negative.
Recalling our previous discussion about gambling systems, this means that we do not allow the player to
have a negative balance. This choice in the definition is essential, as it 
prevents the success of a `martingale betting system' as we described it.
Continuing with our definition of martingales as betting strategies,
we say that 
$f$ succeeds on a real $X$ if
\[
\limsup_{n\to \infty}f(X\upharpoonright n) = \infty.
\]
Schnorr \cite{Schnorr} was interested in an algorithmic concept of
randomness. Incidentally, Martin-L\"of \cite{ML} had already provided a
mathematical definition of randomness based on computability theory and effective measure theory.
But Schnorr wanted to approach this challenge via the intuitive concept of betting strategies.
He
proved that a real (i.e.\ an infinite binary sequence) $X$ is 
Martin-L\"of random if and only if no effective martingale can 
succeed on it.
Here ``effective'' means 
that $f$ is computably approximable from below.
Schnorr's result is an effective version of the maximal inequality for martingales
in probability theory, which says that with probability 1 a non-negative 
martingale is bounded.
There is a huge literature about the relationship between 
martingales and effective randomness, and variations on the theme,
such as computable martingales and randomness, partial computable 
martingales, nonmonotonic martingales, polynomial time 
martingales, etc. We refer the reader to Downey and Hirschfeldt
\cite{roddenisbook} and Nies \cite{Ottobook} for some details
and further background.

\subsection{Why integer-valued martingales?}
Recall the standard criticism of martingale betting systems, i.e.\ that their success
depends on the ability of the player to sustain arbitrarily large losses. This criticism
lead (for the purpose of founding algorithmic randomness) to defining a martingale
as a function from the space of coin-tosses to the non-negative reals
(instead of all the reals) which represent the possible values of the capital available
to the player. There is another criticism 
on such betting strategies that was not taken into account in the formal definition.
Schnorr's definition of a
martingale (as a betting strategy) allows betting infinitesimal 
(i.e.\ arbitrarily small) amounts.
Clearly such an option is not available
in real gambling situations, say at a casino,
where you cannot bet arbitrarily small amounts on some outcome.
It becomes evident that restricting  the betting strategies to a discrete
range results in a more realistic concept of betting.
Such considerations led 
Bienvenu, Stephan and Teutsch 
\cite{BST}
to introduce and 
study \emph{integer-valued martingales}, and the corresponding 
randomness notions. Interestingly, it turns out
that the algorithmic randomness based on
 integer-valued martingales is quite different from
the theory of randomness based on Martin-L\"of \cite{ML} or 
Schnorr \cite{Schnorr} (as developed in the last 30 years, see 
\cite{roddenisbook, Ottobook} for an overview).
The reason for this difference is that most of the classical martingale arguments
in algorithmic randomness make substantial use of the
property of being able to bet infinitesimal amounts (thereby 
effectively avoiding bankruptcy at any finite stage of the process).

Quite aside from the motivations of examining 
the concept of integer-valued martingales 
for its own sake,
if we are to examine the randomness 
that occurs in practice, then such discretised randomness will be the 
kind we would get. The reason is evident:
we can only use a finite number of 
rationals for our bets, and these scale to give 
integer values. \emph{Additionally}, at the more speculative level,
if the universe is granular, finite,  and not a manifold, then if there is 
any randomness to be had (such as in 
quantum mechanics) it will be integer-valued for the same 
reason. 

\begin{table}
\begin{center}
\arrayrulecolor{green!50!black}
  \begin{tabular}{lll}
 \hline\hline
{\small Integer-valued} &	\hspace{0.2cm}  &  {\small Computable}	   \\[1ex]
 {\small Finite-valued}	& \hspace{0.2cm}  &  {\small Single-valued}\\[1ex]
 {\small Partial integer-valued}   & \hspace{0.2cm} & {\small Partial computable}\\[1ex]
\hline\hline
\end{tabular}
\vspace{0.2cm}
\caption{Randomness notions based on effective martingales}
\label{ta:vschelparbouatd}
\end{center}
\end{table}

\subsection{Integer randomness notions and computability}
We formally introduce and discuss the notions of integer-valued randomness
in the context of computability theory.
For the purposes of narrative flow, we will assume that the reader is 
familiar with the basics of algorithmic randomness.
Schnorr
based algorithmic randomness on the concept of effective strategies.
Along with this foundational work, he introduced and philosophically argued for
a randomness notion which is weaker than Martin-L\"of
randomness and is now known as {\em Schnorr randomness}. Further notions, like
{\em computable randomness}, are quite natural from the point of view of betting strategies
and
have been investigated extensively.
Integer-valued martingales induce randomness notions with properties
that are
quite a different in flavour from those 
of, for instance, Martin-L\"of randomness, computable randomness
and the like. 
Our goal in the present paper is to clarify the 
relationship between integer-valued randomness and 
classical degree classes which measure levels of 
computational power.

\begin{defi}[Integer-valued martingales]\label{de:intvalmar}
Given a finite set
$F\subseteq \mathbb{N}$,
we say that a martingale $f$ is $F$-valued if
$f(\sigma i)=f(\sigma)\pm k$ for some $k\in F$. 
A martingale is integer-valued if it is
$\mathbb{N}$-valued, and is 
single-valued if $F=\{a,0\}$ for some $a\ne 0$.
\end{defi}
\noindent
Note that a martingale is $F$-valued if at any stage we can only bet 
$k$ dollars for some $k\in F$ on one of the outcomes $i\in \{0,1\}$, and 
must lose $k$ dollars if $1-i$ is the next bit. 
We note that partial integer-valued martingales are defined as in Definition
\ref{de:intvalmar}, only that the martingales can be partial.
In the following we often say that, given a string $\sigma$, the string $\sigma 0$
is the {\em sibling} of $\sigma 1$ (and $\sigma 1$ is the {\em sibling} of $\sigma 0$).

If we restrict our attention to the countable class of computable or
partial computable martingales, we obtain a number of 
algorithmic randomness notions.
For example, a real is [partial] computably random if no [partial] computable
martingale succeeds on it. Similar notions are obtained if we consider
integer-valued martingales. 

\begin{defi}[Integer-valued randomness]
A real $X$ is [partial] integer-valued random  
if no [partial] computable integer-valued martingale succeeds on it.
Moreover $X$ is finitely-valued random 
if for each finite set $F\subseteq\mathbb{N}$, 
no computable $F$-valued martingale succeeds on it, and is 
single-valued random 
if no computable single-valued martingale succeeds on it.
\end{defi}
\noindent
We list these randomness notions in Table \ref{ta:vschelparbouatd}, along
with the traditional randomness notions computable and partial computable randomness.
Note that partial integer-valued randomness is stronger than integer-valued randomness,
just as partial computable randomness is stronger than computable randomness.
Bienvenu, Stephan and Teutsch \cite{BST} 
clarified the relationship between integer-valued, single-valued and 
a number of other natural randomness notions. 
Figure \ref{fig:dynam2} illustrates some of the implications that they
obtained.
We already know (see \cite{roddenisbook})  that 
 computable randomness implies 
Schnorr randomness, which in turn implies Kurtz randomness
(with no reversals)
and that  Schnorr randomness implies the law of large numbers.
Bienvenu, Stephan and Teutsch proved that 
if we add the above notions to the diagram in Figure \ref{fig:dynam2},
no other implications hold apart from the ones
in the diagram and the ones just mentioned.
In addition, we may add a node for `partial computably random'
and an arrow from it leading to the node `computably random'.
Nies showed in \cite[Theorem 7.5.7]{Ottobook} 
that the converse implication does not hold, 
i.e.\ there are computably random reals which are not partial computably
random. 
A strong version of this fact holds for integer-valued randomness.
We show that there are
integer-valued random reals which not only are not partial integer-valued random, but
they do not contain any partial integer-valued random reals in their Turing degree.

\begin{figure}
\begin{tikzpicture}[
Nnode/.style={rectangle,  minimum size=6mm, very thick,
draw=green!50!black!50, top color=white, rounded corners, bottom color=green!50!black!20, font=\small },
tnode/.style={rectangle,  minimum size=6mm, rounded corners, very thick,
draw=red!50!black!50, top color=white, bottom color=red!50!black!20, 
 font=\small},
 ttnode/.style={rectangle,  minimum size=6mm, rounded corners, very thick,
draw=yellow!50!black!50, top color=white, bottom color=yellow!50!black!20, 
 font=\small}]
 \node (cr) [Nnode, outer sep=3pt, inner sep=7pt] at (1,1.5) {\textrm{computably random}};
\node (ivr) [Nnode, outer sep=3pt, inner sep=7pt] at (5.5,1.5) {\textrm{\ integer-valued random\ \ }};
\node (fvr) [Nnode, outer sep=3pt, inner sep=7pt] at (5.5,0) {\textrm{\ finitely-valued random\ \ }};
\node (svr) [Nnode, outer sep=3pt, inner sep=7pt] at (10,0) {\textrm{\ single-valued random\ \ }};
\draw [->] (cr) -- (ivr); \draw [->] (ivr) -- (fvr); \draw [->] (fvr) -- (svr);
\node (kr) [tnode, outer sep=3pt, inner sep=7pt] at (10,1.5) {\textrm{\ Kurtz random\ \ }};
\node (bi) [ttnode, outer sep=3pt, inner sep=7pt] at (1,0) {\textrm{\ bi-immune\ \ }};
\draw [->] (kr) -- (svr); 
\draw [->] (fvr) -- (bi); 
\end{tikzpicture}
\caption{{\textrm Implications between randomness notions obtained
in \cite{BST}.}}
\label{fig:dynam2}
\end{figure}
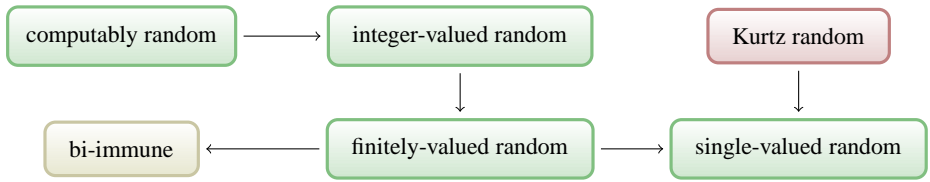

One interesting observation of Bienvenu, Stephan and Teutsch 
\cite{BST}, was that integer-valued
was a meeting point of genericity (and hence category) and 
measure since weakly $2$-generic sets are 
\ivr. Hence the \ivrs\ are co-meager as well as having
measure 1. The reader should recall that 
a subset of $\mathbb{N}$ is called $n$-generic if it 
meets or avoids all $\Sigma_n^0$ sets of strings, and is 
weakly $n$-generic if it meets all dense $\Sigma_n^0$ sets of strings.
The reason that highly generic reals are \ivr\ is that there is a finitary strategy 
to make a set \ivr, since we can force an opponent who bets to lose.
The point here is that if the minimum bet is  
one dollar and he has $k$ dollars to spend, then 
he can lose at most $k$ times, so we can use a finite strategy 
to force the opponent into a cone (in Cantor space) where he cannot win.
This finite strategy is not available if arbitrarily small bets are allowed.
Naturally the question 
arises as to what level of genericity is needed for constructing \ivrs.
Bienvenu, Stephan and Teutsch
\cite{BST} proved that it is possible to have a $1$-generic which
was \emph{not} \ivr.
So the answer they gave to this question is `somewhere between 
weak 2 and 1-genericity'.
As we see in the next section, we give a more precise answer.

\subsection{Our results, in context}
Bienvenu, Stephan and Teutsch 
\cite{BST} showed that the class of integer-valued random is co-meager, so
sufficient genericity is a guarantee for this kind of randomness. They also 
quantified this statement via the hierarchy of genericity,
showing that
the genericity required lies 
somewhere between weak 2-genericity and 1-genericity.
We show that a notion of genericity from \cite{DJS2}
which is known as $pb$-genericity, 
implies (partial) integer-valued randomness.
Recall from \cite{DJS2} that set of strings $S$ is $pb$-dense if
it contains the range of a  function $f:2^{<\omega}\to 2^{<\omega}$ with a computable approximation
$(f_s)$ such that $f(\sigma)\succcurlyeq \sigma$ for all strings $\sigma$
and $| \{ s\ |\  f_{s+1}(\sigma) \ne f_{s}(\sigma) \}| \leqslant h(\sigma)$ 
where the function $h: 2^{<\omega}\to \omega$
is primitive recursive. A real $X$ is $pb$-generic if every $pb$-dense
set of strings contains a prefix of $X$.

\begin{thm}[Genericity for integer-valued randoms]\label{wMGIVR}
Every $pb$-generic real is (partial) \ivr. 
\end{thm}
\noindent
This result might suggest that integer-valued randomness and partial
integer-valued randomness are not easily distinguishable. 
In Section \ref{subse:anivrnotpivr} we present a rather elaborate 
finite injury construction
of a real which is integer-valued random but not partial \ivr.
In Section \ref{subse:invrncomppivr} this construction is modified into a
$\mathbf{0}''$ tree argument, which proves the following degree separation
between the two randomness notions.

\begin{thm}[Degree separation of randomness notions]\label{th:pivrfromivr}
There exists an integer-valued random $X<_T \mathbf{0}'$
which does not compute any partial
integer-valued random.
\end{thm}

We are interested in classifying the computational power that is associated with
integer-valued randomness.
Computational power is often represented by properties of degrees, which
in turn define classes like the degrees which 
can solve the halting problem, or the array non-computable degrees
from \cite{DJS2}.
The reader might recall that $A$ is array noncomputable 
if and only if for all $f\le _{wtt}\emptyset'$,
 there is a function
$g\le_T A$ such that $g(n)>f(n)$ for infinitely many $n$. 
This class has turned out to be ubiquitous 
and characterized classes defined by many distinct combinatorial 
properties. 
Recall that a presentation of a real $A$ is a
 c.e.\ prefix-free set of strings  representing an open set of 
 Lebesgue measure $\alpha$.
The c.e.\ array noncomputable degrees are exactly the c.e.\ 
degrees that 
\begin{itemize}
\item[(a)] contain c.e.\ sets $A$ of infinitely often maximal (i.e.\ $2\log n$)
Kolmogorov complexity; (Kummer \cite{kummer})
\item[(b)] have effective packing dimension 1; (Downey and Greenberg \cite{DGpacking})
\item[(c)] compute left-c.e.\ reals 
$B<_T A$ such that 
every presentation of $A$ is computable from $B$; (Downey
and Greenberg \cite{DGpactransf})
\item[(d)] compute a pair of disjoint c.e.\ sets such that
 every separating set for this pair
computes the halting problem; (Downey, Jockusch, and Stob \cite{DJS1})
\item[(e)] do not have strong minimal covers; (Ishmukhametov \cite{Ishm})
\end{itemize}
Also by Cholak, Coles, Downey, Herrmann 
\cite{cohere}
the array noncomputable
c.e.
degrees form an invariant class for the lattice of $\Pi_1^0$
classes via the thin perfect classes. 

Theorem \ref{wMGIVR} can be used to 
show that a large class of degrees compute (partial)
\ivrs. By \cite{DJS2}, every \anc\ degree computes a $pb$-generic.
Therefore Theorem \ref{wMGIVR} has the following consequence.

\begin{coro}[Computing integer-valued randoms]\label{coro:anccpivr}
Every array noncomputable degree computes a (partial)
integer-valued random.
\end{coro}
\noindent
Note however that an \ivr\ need not be of \anc\ degree. Indeed,
it is well known that there are array computable \ml randoms.
A converse of Corollary \ref{coro:anccpivr} can be obtained
for the c.e.\ degrees, as Theorem \ref{ancccc} shows.

While genericity is an effective tool for exhibiting integer-valued randomness
in the global structure of the degrees (as we demonstrated above)
it is incompatible with computable enumerability.
Since generic degrees (even 1-generic) are not c.e., investigating
integer-valued randomness in the c.e.\ degrees requires a different analysis.
Already the fact that randomness can be exhibited in the c.e.\ degrees is
quite a remarkable phenomenon, and restricted to weak versions of randomness.
Martin-L\"of randomness is the strongest standard randomness notion that
can be found in the c.e.\ degrees. A c.e.\ degree contains a  Martin-L\"of random set
if and only if it is complete (i.e.\ it is the degree of the halting problem).
Furthermore, the complete c.e.\ degree contains the most well-known 
random sequence---Chaitin's $\Omega$---which is
the measure of the domain of a universal prefix-free machine, the 
universal \emph{halting probability}. 
An interesting characteristic of this random number is that it is left-c.e., i.e.\ it
can be approximated by a computable increasing sequence of rationals. 
Weaker forms of randomness---like computable and Schnorr randomness---can
be found in incomplete c.e.\ degrees, even in the form of left-c.e.\ sets.
For example, Nies, Stephan and Terwijn \cite{NST} showed that the c.e.\ degrees
which contain computably and Schnorr random sets are exactly the high c.e.\ degrees.
Moreover each high c.e.\ degree contains a 
computably random left-c.e.\ set 
and a Schnorr random left-c.e.\ set.
We prove an analogous result for integer-valued randomness and partial
integer-valued randomness.

\begin{thm}[C.e.\ degrees containing \ivr\ left-c.e.\ sets] \label{cedegcontivrlceres}
A c.e.\ degree contains a (partial) left-c.e.\ \ivr\ if
and only if it is high.
\end{thm}
This result is pleasing, but this is where
the similarities between computable randomness (based on computable betting strategies) 
and integer-valued randomness (based on integer-valued computable betting strategies) end, 
at least with respect to the c.e.\ degrees.
We note that 
\begin{equation}\label{eq:cedcompschl}
\left\{\parbox{11cm}{\ \ In the c.e.\ degrees, the following classes 
are equal to the high degrees:
\begin{itemize}
\item[(i)] degrees containing computably random sets; 
\item[(ii)] degrees containing left-c.e.\ computably random sets;
\item[(iii)] degrees computing computably random sets.
\end{itemize}}\right.
\end{equation}
This characterization follows from the following facts, where
(a) is by Schnorr  \cite{Schnorr}, (c) was first observed by Ku\v{c}era \cite{MR820784}, and
(b), (d) are from \cite{NST}.
\begin{itemize}
\item[(a)] computable randomness implies Schnorr randomness; 
\item[(b)] a Schnorr random which does not have high degree
is Martin-L\"of random;
\item[(c)] a Martin-L\"of random of c.e.\ degree is complete;
\item[(d)] every high c.e.\ degree contains a computably random left-c.e.\ set.
\end{itemize}
In the case of integer-valued randomness \eqref{eq:cedcompschl} fails significantly.
In particular, the c.e.\ degrees that compute integer-valued randoms are not the same as the
c.e.\ degrees that contain integer-valued randoms.
In fact, we provide the following characterization of the c.e.\ degrees that compute \ivrs.

\begin{thm}[C.e.\ degrees computing \ivrs]\label{ancccc}
A c.e.\ degree  
computes an \ivr\  
if and only if it is array noncomputable.
\end{thm}
\noindent
In view of this result
it is tempting to think that that every c.e.\ array noncomputable
might contain an \ivr.
We will see however 
in Section \ref{subse:anrncontivr}
that there are array noncomputable c.e.\ degrees which do not contain \ivrs.
In fact, Theorem \ref{th:high2cenoivr} is an extreme version of this fact, which is tight with
respect to the jump hierarchy.

We have seen that high c.e.\ degrees are powerful enough to
contain \ivrs, even left-c.e.\ \ivrs. 
However, while we know that left-c.e.\ \ivrs\ necessarily have high degree,
the question arises as to whether a weaker jump class is sufficient to
guarantee that a c.e.\ degree contains an \ivr, if we no longer require that the set is
left-c.e. We give a negative answer to this question by the following result, which we
prove in Section \ref{subse:high2nocontivrs} by a ${\bf 0'''}$-argument. 

\begin{thm}\label{th:high2cenoivr}
There is a high$_2$ c.e.\ degree which does not contain any \ivrs.
\end{thm}
\noindent
Note that this result shows the existence of array noncomputable c.e.\ degrees
which do not contain \ivrs, in stark contrast to Theorem \ref{ancccc}.

Furthermore, the c.e.\ degrees that contain integer-valued randoms are
not the same as the c.e.\ degrees that contain left-c.e.\ integer-valued randoms.
In fact, in stark contrast with Theorem \ref{cedegcontivrlceres},
there exists a low c.e.\ degree containing an integer-valued random set.
More generally, we can find c.e.\ degrees containing integer-valued random sets
in every jump class.
Section \ref{se:jumpinvivrs} is devoted to the proof of this result.

\begin{thm}[Jump inversion for c.e.\ \ivr\ degrees] \label{th:ceaiv}
If ${\bf c}$ is c.e.\ in and above $\mathbf{0}'$ then 
there is an \ivr\ $A$ of c.e.\ degree with 
$A'\in {\bf c}$.
\end{thm}

There are a number of open questions and research directions pointed
by the work in this paper. For example, is there a c.e.\ degree which contains
an integer-valued random
but does not contain any partial integer-valued randoms?
More generally, which degrees contain partial computable randoms? 
Algorithmic randomness based on
partial martingales is a notion that remains to be explored on a deeper level.

\section{Genericity and partial \ivrs}
Bienvenu, Stephan, Teutsch \cite[Theorem 8]{BST} showed  
that every 
weakly 2-generic set is \ivr.
In Section \ref{subse:profthgenimppivr}
we give a proof of Theorem \ref{wMGIVR}, i.e.\ that 
$pb $-genericity is sufficient for (partial) integer-valued randomness.

Hence a certain notion of genericity ($pb$-genericity) is
a source of integer-valued randomness.
In fact, by Theorem  \ref{wMGIVR}, every $pb$-generic is not only \ivr\ but also \pivr.
We do not have concrete examples of reals that are \ivr\ but not \pivr. 
Section \ref{subse:anivrnotpivr} is dedicated to
constructing such an example. 
We give the basic construction of a $\Delta^0_2$
real which is \ivr\ but not \pivr. We give this in full detail, as it
is based on an interesting idea.

In Section \ref{subse:invrncomppivr} we provide
the necessary modification of the previous construction
in order to show that
the degrees of \ivrs\ and \pivrs\ can also be separated, even inside $\Delta^0_2$. 
In particular, we are going to prove
Theorem \ref{th:pivrfromivr}, i.e.\ 
that there is a $\Delta^0_2$ \ivr\ which does not compute any \pivrs.
These modifications are essentially the implementation of the main
argument of Section \ref{subse:anivrnotpivr} on a tree, which results from the additional requirements that introduce infinitary outcomes that need to be guessed
(i.e.\ the totality of the various functionals with oracle the constructed set). 
Given the original strategies and construction, the tree argument is
fairly standard.

\subsection{Proof of Theorem \ref{wMGIVR}}\label{subse:profthgenimppivr}
For every partial computable integer-valued martingale $m$ with
effective approximation $(m_s)$
we define a function $\hat{m}: 2^{<\omega}\to 2^{<\omega}$ with uniformly 
computable approximation $(\hat{m}_s)$. 
Let $\hat{m}_0(\sigma) = \sigma$ for all strings $\sigma$. 
Inductively in $s$, suppose we have defined $\hat{m}_s$. 
At stage $s+1$, if $m_{s+1}(\hat{m}_s(\sigma))$ is
 defined and there exists an extension $\tau$ of it of length $\leq s+1$
such that $m_{s+1}(\tau)$ is defined and
$m_{s+1}(\tau)< m_{s+1}(\hat{m}_s(\sigma))$, then we define
$\hat{m}_{s+1}(\sigma)$ to be the least such string $\tau$ (where strings are ordered first
by length and then lexicographically). Otherwise
let $\hat{m}_{s+1}(\sigma) = \hat{m}_{s}(\sigma)$.

Since $m$ is an integer-valued martingale
we have 
$| \{ s\ |\  \hat{m}_{s+1}(\sigma) \ne \hat{m}_s(\sigma) \}\ |\  \leqslant m(\sigma) +1$
and $m(\sigma)\leq 2^{|\sigma|}\cdot m(\emptyset)$.
Note that given any partial computable martingale $m$, 
the range of $\hat{m}$ is dense.
So every 
$pb$-generic intersects the range of $\hat{m}$ for every partial computable martingale $m$.
Moreover given any partial computable martingale $m$, 
the range of $\hat{m}$ is a subset of
\[
W_m = \{ \sigma\ |\  m(\sigma') \simeq m(\sigma) \text{ for all extensions } \sigma' \text{ of } \sigma \}.
\]
So every 
$pb$-generic intersects $W_m$ for each partial computable 
martingale $m$.
This means that every $pb$-generic is partial \ivr.

\subsection{An \ivr\ which is not \pivr}\label{subse:anivrnotpivr}
It suffices to construct an \ivr\ set $A\leq_T \emptyset'$ and a partial integer-valued martingale
$m$ which succeeds on $A$.
We will define a computable approximation $(A_s)$ which converges to a set $A$ which has 
the required properties.
Let $\langle n_e \rangle_{e \in \omega}$ be an effective list of all partial computable integer-valued martingales.
For each $e \geqslant 2$ we need to satisfy the requirements

\begin{description}
	\item[$R_e$]  If $n_e$ is total, then $n_e$ does not succeed on $A$.
	\item[$Q_e$]  $m$ wins  (at least)  \$$e$ on $A$.
\end{description}
\noindent
We first see how we might meet one requirement $R_0$. We begin by setting $A_0 = 1^{\omega}$ and defining $m$ to start with \$2 and wager \$1 on every initial segment of $A_0$. If we later see that $n_0$ has increased its capital along $A_0$, then we would like to move our approximation to $A$ so that $n_0$ decreases in capital. In changing $A_s$ to decrease $n_0$'s capital we may also decrease $m$'s capital along $A_s$. If $n_e$ had, say, \$10 in capital at some point, then as it loses at least a dollar every time it decreases in capital, it can lose at most 10 times. Our martingale, if it bets in \$1 wagers, can withstand losing \$10 only if its capital at that point is at least \$11. 

If $m$ does not have sufficient capital for us to start attacking immediately, we must find a way to increase its capital. We can increase $m$'s capital to \$$k$ as follows. We have not yet defined $m$ on any string extending $0$. We therefore wait until $n_0$ has halted on all strings of length $k$. If this never happens, then $n_0$ is not total, and $R_0$ is met. If $n_0$ does halt on all strings of length $k$, we pick a string $\tau$ of length $k$ extending $0$ such that $n_0(\tau) \leqslant n_0(0)$. Such a string must exist since
$n_0$ is a (partial) martingale. We are then free to define $m$ to wager \$1 on every initial segment of $\tau$ and set $A_s = \tau \ \hat{\empty} \ 1^{\omega}$. We will then have $m(\tau) > n_0(\tau)$. If $n_0$ later increases its capital along $A_s$ we will be able to change the approximation to $A$ to decrease $n_0$'s capital. As $m$ now has greater capital than $n_0$, we will be able to decrease $n_0$'s capital to \$0 while ensuring that $m$ does not run out of money.

When dealing with multiple requirements, we must take care in defining $m$ as it is a global object. We set a restraint $r_e$ for every $e \in \omega$. We arrange things so that only $R_e$ will be able to define $m$ on a string extending $A_s \upharpoonright r_{e,s} \ \hat{\empty} \ 0$. Suppose $R_e$ has not required attention since it was last injured. When we see $n_e$ increase its capital above $n_e(A_s \upharpoonright r_e)$, rather than starting to attack immediately, even if possible, we choose a string $\tau'$ extending $A_s \upharpoonright r_{e,s} \ \hat{\empty} \ 0$  
and define $m$ such that $m(\tau') - m(A_s \upharpoonright r_{e,s} \ \hat{\empty} \ 0) > n_e(A_s \upharpoonright r_{e,s} \ \hat{\empty} \ 0)$.
We injure requirements of weaker priority by lifting their restraints. We may then decrease $n_e$'s capital to \$0 and still have $m$ left with some capital. We now turn to the formal details of the construction.

We have for every requirement $R_e$ a restraint $r_e$. At every stage a requirement will either be declared to be waiting for convergence at some length, or declared to not be waiting for convergence at any length. As usual, this will stay in effect at the next stage unless otherwise mentioned. We say that $R_e$ requires attention at stage $s$ if either 

\begin{enumerate}
\item[(i)]
$R_e$ was declared to not be waiting for convergence at any length at stage $s$, and there is $l$ such that
\begin{enumerate}
 \item $l > r_{e,s}$, 
\item $n_{e,s}(\sigma)\downarrow$ for all strings $\sigma$ of length $l$, and 
\item  $n_e(A_s \upharpoonright l) > n_e(A_s \upharpoonright (l-1))$, or
\end{enumerate}
\item[(ii)] $R_e$ was declared to be waiting for convergence at length 
$h$ at stage $s$, and $n_{e,s}(\sigma) \downarrow$ for all strings $\sigma$ of length $h$.
\end{enumerate}
\noindent
In Case (i) we say that $R_e$ requires attention through $l$. We say that $Q_e$ requires attention at stage $s$ if $m(A_s \upharpoonright r_{e,s}) < e$. We order the requirements as $R_0, R_1, R_2, Q_2, R_3, Q_3, \dots$.

\subsection*{Construction}

\textit{Stage 0}: Set $A_0 = 1^{\omega}$ and $m(\lambda) = 2$, and let $m$ wager \$1 on every initial segment of $A_0$. Set $r_{e,0} = e$ for all $e \in \omega$. For all $e \in \omega$, declare that $R_e$ is not waiting for convergence at any length at stage $1$.

\textit{Stage $s$, $s \geqslant 1$}: Find the requirement of strongest priority which requires attention at stage $s$. (If no such requirement exists, go to the next stage.) There are several cases.

\textit{Case 1}: $R_e$ requires attention at stage $s$ in Case (i).
Has $R_e$ required attention since it was last injured?

\textit{Subcase 1a}: 
No.
Declare $R_e$ to be waiting for convergence at length 
$n_e((A_s \upharpoonright r_{e,s}) \ \hat{\empty} \ 0) + 1 + (r_{e,s} + 1)$ at stage $s+1$.

\textit{Subcase 1b}: 
Yes. 
Suppose $l$ is least such that $R_e$ requires attention through $l$ at stage $s$.
Let $A_{s+1} = (A_s \upharpoonright (l-1)) \ \hat{\empty} \ (1- A_s(l-1)) \ \hat{\empty} \ 1^{\omega}$. Define $m$ to wager \$1 on every initial segment of $A_{s+1}$ of length at least $l + 1$. For all $e' > e$, let $r_{e',s+1}$ be a fresh large number such that for all $e_1 < e_2$ we have $r_{e_1,s+1} < r_{e_2, s+1}$, and declare that $R_{e'}$ is not waiting for convergence at any length at stage $s+1$.

\textit{Case 2}: $R_e$ requires attention at stage $s$ in Case (ii). Suppose $R_e$ was declared to be waiting for convergence at length $h$ at stage $t$.

Choose a string $\tau$ above $A_s \upharpoonright r_{e,s} \ \hat{\empty} \ 0$ of length $h$ such that $n_e(\tau) \leqslant n_e(A_s \upharpoonright r_{e,s} \ \hat{\empty} \ 0)$. Define $m$ to wager \$1 along every initial segment of $\tau$ with length in $(r_{e,s} + 1, |\tau|]$. Set $A_{s+1} = \tau \ \hat{\empty} \ 1^{\omega}$. Define $m$ to wager \$1 on every initial segment of $A_{s+1}$ of length at least $|\tau| + 1$. For all $e' > e$, let $r_{e',s+1}$ be a fresh large number such that if $e_1 < e_2$ then $r_{e_1,s+1} < r_{e_2, s+1}$. For all $e'' \geqslant e$, declare that $R_{e''}$ is not waiting for convergence at any length at stage $s+1$.

\textit{Case 3}: $Q_e$ requires attention at stage $s$. Let $\tau$ be the least string extending $A_s \upharpoonright r_{e,s}$ for which $m(\tau) = e$. Set $A_{s+1} = \tau \ \hat{\empty} \ 1^{\omega}$ and for all $e' > e$, let $r_{e',s+1}$ be a fresh large number such that if $e_1 < e_2$ then $r_{e_1,s+1} < r_{e_2, s+1}$. For all $e'' > e$, declare that $R_{e''}$ is not waiting for convergence at any length at stage $s+1$.

\subsection*{Verification}
Before we demonstrate the satisfaction of requirements $R_e, Q_e$, we need to show that
the partial martingale $m$ is well-defined, and for all $n, s \in \omega$, 
\begin{equation}\label{eq:misweldfnz}
\textrm{if\hspace{0.3cm} $m(A_s \upharpoonright n)\de$\hspace{0.3cm} then\hspace{0.3cm} $m(A_s \upharpoonright n) \geqslant 1$.}
\end{equation}
We clarify that in this statement, $m$ denotes the state of the partial martingale at stage $s$.
We prove this statement by induction on the stages $s$.
We first claim that if $R_e$ requires attention in Case (ii) at stage $s$, then $m$ has not been defined on any string extending $A_s \upharpoonright r_{e,s} \ \hat{\empty} \ 0$. Suppose $R_e$ requires attention in Case (ii) at stage $s$. 
Let $s^* -1 < s$ be the last stage at which $R_e$ was initialised. We choose $r_{e,s^*}$ to be some fresh large number. In particular, $m$ has not been defined on any string extending $A_{s^*} \upharpoonright r_{e,s^*} \ \hat{\empty} \ 0$. 
Note that $r_{e, s} = r_{e,s^*}$ and $A_{s^*} \upharpoonright r_{e,s} = A_s \upharpoonright r_{e,s}$. If $R_k$ for $k > e$ acts at stage $t > s^*$ it may define $m$ on $A_t \upharpoonright r_{k,t} \ \hat{\empty} \ 0$, but as $r_{k,t} > r_{e,t} \geqslant r_{e,s}$, it cannot define $m$ on any string extending $A_{s^*} \upharpoonright r_{e,s} \ \hat{\empty} \ 0$. 
No $R_i$ for $i < e$ may act between stages $s^*$ and $s$ as this contradicts the choice of $s^*$. 
Therefore $m$ has not been defined on any string extending $A_s \upharpoonright r_{e,s} \ \hat{\empty} \ 0$ at stage $s$.

By the definition of $A_0$ and $m$ at stage 0, we have $m(A_0 \upharpoonright n) \geqslant 2$ for all $n$. 
Furthermore, $m$ is at least 1 on the sibling of any initial segment of $A_0$ (recall the definition of the {\em sibling} of a string, just after 
Definition \ref{de:intvalmar}).
Suppose that $m(A_s \upharpoonright n) \geqslant 1$ for all $n$ and $m$ is at least 1 for any sibling of an initial segment of $A_s$.
If we act for $Q_e$ at stage $s$, then it is easy to see that $m(A_{s+1} \upharpoonright n) \geqslant 1$ for all $n$.
So suppose that we act for $R_e$ at stage $s$. If we act in Case 2 for $R_e$ at stage $s$, then we will choose a string $\tau$ extending $A_s \upharpoonright r_{e,s} \ \hat{\empty} \ 0$ of some length $h$. By assumption, $m(A_s \upharpoonright r_{e,s} \ \hat{\empty} \ 0) \geqslant 1$.
We then let $A_{s+1}$ extend $\tau$ and define $m$ such that $m(\tau) - m(A_s \upharpoonright r_{e,s} \ \hat{\empty} \ 0) > n_e(A_s \upharpoonright r_{e,s} \ \hat{\empty} \ 0)$. Then $m(A_{s+1} \upharpoonright n) \geqslant 1$ for all $n$. 

Now suppose that we act for $R_e$ in Subcase 1b through length $l$ at stage $s$. Our martingale $m$ wagers at most \$1 at a time, and so loses at most \$1 at a time. We decrease $n_e$ by at least \$1 while decreasing $m$ by at most \$1. As $m(A_s \upharpoonright l) > n_e(A_s \upharpoonright (l-1))$, we may reduce $n_e$'s capital to \$0 while $m$ has capital remaining. Now requirements of stronger priority than $R_e$ may start to act. Suppose that $R_{e'}$ with $e' < e$ requires attention. If $R_{e'}$ requires attention in Case (ii) then we will act as in the previous paragraph and $m$ will still have capital left. Otherwise, $R_{e'}$ may act in Subcase 1b at stage $t$ after having acted in Case 2 \textit{before} stage $s$. However, in this case, we would have increased $m$'s capital by $n_{e'}(A_s \upharpoonright r_{e',s} \ \hat{\empty} \ 0)$ previously. Therefore, after having reduced $n_e$'s capital to 0, we may then reduce $n_{e'}$'s capital to 0 as well, while
ensuring that  $m$ still has capital remaining.
This concludes the induction on the stages and the proof of \eqref{eq:misweldfnz}.

Note that we have not yet shown that the approximation $(A_s)$ converges to a set $A$.
This is a consequence of the use of restraints in the construction, and the following lemma
which says that each requirement $R_e$ receives attention only finitely often.

\begin{lem}
For all $e \in \omega$, $R_e$ receives attention only finitely often, and is met.
\end{lem}
\begin{proof}
Suppose by induction that $s^*-1$ is the last stage at which $R_e$ is injured
(i.e.\ the least stage after which no requirement of stronger priority than $R_e$ receives attention).
If $R_e$ never requires attention at some later stage in Case (i), then either $n_e$ is not total, or $n_e(A) \leqslant n_e(A_{s^*} \upharpoonright r_{e,s^*})$. In either case $R_e$ is met.
Therefore suppose that $R_e$ requires attention through $l$
at some stage $s' \geqslant s^*$. 
We will act in Subcase 1a and declare $R_e$ to be waiting for convergence at length $n_e((A_{s'} \upharpoonright r_{e,s'}) \ \hat{\empty} \ 0) + 1 + (r_{e,s'} + 1) =: h_0$.
As no requirement $R_{e'}$ for $e' < e$ receives attention after stage $s^*$, $R_e$ will be waiting for convergence at length $h_0$ until, if ever, $R_e$ requires attention in Case (ii). If $R_e$ never requires attention after stage $s'$ then $n_e$ is not total. So suppose that $R_e$ requires attention in Case (ii) at stage $t'$. We choose a string $\tau$ of length $h_0$ above $A_{t'} \upharpoonright r_{e} \ \hat{\empty} \ 0$ such that $n_e(\tau) \leqslant n_e(A_{t'} \upharpoonright r_e \ \hat{\empty} \ 0)$. Since $n_e$ is a (partial) martingale,  $n_e(\sigma) \leqslant n_e(A_{t'} \upharpoonright r_e \ \hat{\empty} \ 0)$ for at least one string $\sigma$ of length $h_0$ above $A_{t'} \upharpoonright r_{e} \ \hat{\empty} \ 0$. Therefore such a string must exist.

We set $A_{t'+1} = \tau \ \hat{\empty} \ 1^{\omega}$ and define $m$ so that $m(\tau) > n_e(\tau)$. If $R_e$ receives attention after stage $t'$ then it must do so in Subcase 1b. Our martingale $m$ wagers at most \$1 at a time, and so loses at most \$1 at a time. If $R_e$ requires attention through some $l > h_0$ at a stage $t'' > t'$ then we will again act in Subcase 1b and force $n_e$ to lose at least \$1 while $m$ loses at most \$1. This can happen at most $n_e(\tau)$ many times before $n_e$ loses all its capital and can no longer bet. Thus the induction can continue, and $R_e$ is met.
\end{proof}
It remains to show that $m$ succeeds on $A$. For this, it suffices to show that all 
requirements $Q_e$, $e \geqslant 2$ are met.

\begin{lem}
For all $e \geqslant 2$, $Q_e$ receives attention only finitely often, and is met.
\end{lem}
\begin{proof}
Suppose by induction that $s^*$ is the least stage after which no requirement of stronger priority than $Q_e$ receives attention. As no requirement of stronger priority than $Q_e$ receives attention after stage $s^*$, the restraint $r_{e,s^*}$ will never again be increased unless $Q_e$ acts. If $Q_e$ requires attention at stage $s > s^*$ then we must have $m(A_s \upharpoonright r_{e,s}) < e$. As no requirement of stronger priority receives attention after stage $s^*$, we must have that $Q_{e-1}$ is satisfied, and so $m(A_s \upharpoonright r_{e,s}) = e-1$. We have defined $m$ to wager \$1 on all initial segments of $A_s$ and so there is $\tau$ such that $m(\tau) =e$. We let $A_{s+1} = \tau \ \hat{\empty} \ 1^{\omega}$ and increase the restraints $r_{e',s+1}$ for all $e' \geqslant e$. We then have that $\tau \prec A$ and $Q_e$ is satisfied. 
\end{proof}

\subsection{Integer-valued randoms not computing \pivrs}\label{subse:invrncomppivr}
The construction of Section \ref{subse:anivrnotpivr}, non-trivial as it is, admits some modifications. For example, it is not
hard to add the requirement that $A$ is 1-generic and still successfully perform the argument. This requirement can be
canonically split into an infinite sequence of conditions, with corresponding strategies in the constructions which
will occasionally change the approximation to $A$. Since the 1-genericity sub-requirements are finitary, their effect will
be similar to (in fact, more benign than) the $R_e$ requirements. 
\[
\textrm{\em There is a 1-generic \ivr\ which is not \pivr.}
\]
Note that 1-generics are generalized low, so since $A$ is 
$\Delta^0_2$, it follows that 
\[
\textrm{\em There is a low \ivr\ which is not \pivr.}
\]
The next modification of the construction of Section \ref{subse:anivrnotpivr}
results in the proof of Theorem \ref{th:pivrfromivr} and 
requires more explanation.
We can replace the requirements $Q_e$ with
\[
\textrm{$Q^{\ast}_{e,k}$:\ \ If $\Phi^A_e$ is total and non-computable then $m$ wins (at least) \$k on $\Phi^A_e$.}
\]
Note that now $m$ will bet on $\Phi^A_e$ rather than $A$. For this reason, the family
of requirements $Q^{\ast}_{e,k}$ need to act under the hypothesis that  $\Phi^A_e$ is total.
This means that we need to implement the argument of Section \ref{subse:anivrnotpivr}
on a tree, where the family
of requirements $Q^{\ast}_{e,k}$ lies below a `mother-node' $Q^{\ast}_e$ which has two outcomes, 
a $\Pi^0_2$ outcome $i$ and a  $\Sigma^0_2$ outcome $f$.
The  outcome $i$ corresponds to the fact that $\Phi_e^A$ has infinitely many expansionary stages
(i.e.\ stages where the least $n$ such that $\Phi_e^A(n)$ is undefined is larger than every before) while
outcome $f$ corresponds to the negation of this statement. Moreover the construction 
guarantees that if $i$ is a true outcome, then $\Phi_e^A$ is total. 
Requirements $Q^{\ast}_{e,k}$ act as  the $Q_k$ of Section
\ref{subse:anivrnotpivr} while $R_e$ are the same in the two constructions. Moreover these two requirements
have a single outcome in the tree argument. 
The crucial point here is that if $\Phi^A_e$ is total and non-computable then $\Phi_e$ will have splitting
along $A$, i.e.\ for each prefix $\tau$ of $A$ there will be two finite extensions $\tau_i$ of $\tau$ and an argument $x$
such that $\Phi_e^{\tau_0}(x)\neq\Phi_e^{\tau_1}(x)$. This means that before strategy $Q^{\ast}_{e,k}$ starts operating,
it can secure a  splitting which it can use to move away from versions of $\Phi^A_e$ on which $m$ has not bet appropriately.
In the construction of Section \ref{subse:anivrnotpivr} this happened automatically as $m$ bet on the real itself, and not its image
under a Turing functional. Other than these points, the construction and verification is entirely similar to that of Section \ref{subse:anivrnotpivr}.
Since this there is no novelty in this extension of the argument
of Section \ref{subse:anivrnotpivr} (given the standard machinery for tree arguments and the above remarks)
we leave the remaining details to the motivated reader.

\section{Computable enumerability and integer-valued randomness}
Nies, Stephan and Terwijn \cite{NST} showed that
 a c.e.\ degree is high if and only if it contains a computably random
 c.e.\ real. Moreover an analogous statement holds for partial computable random c.e.\ reals.
 In Section \ref{subse:delceintvr} we show that the same is true for integer-valued random c.e.\ reals.
 In other words, we give the proof of the first part of Theorem \ref{cedegcontivrlceres} that we discussed in the introduction.
The proof of the remaining part of Theorem \ref{cedegcontivrlceres} (regarding partial integer-valued randoms)
is deferred to  Section \ref{subse:delceintv2}, since the required machinery is similar to the one we use for the jump inversion 
  theorems.
In Section \ref{subse:ancdegncomivr} we give the proof of Theorem \ref{ancccc}. Note that by
Corollary \ref{coro:anccpivr}, for this proof it suffices to show that array computable c.e.\ degrees are not
\ivr.

\subsection{Degrees of left-c.e.\ integer-valued randoms}\label{subse:delceintvr}
In this section we prove the  
first part of Theorem \ref{cedegcontivrlceres}, i.e.\ that the high c.e.\ degrees are exactly those c.e.\
degrees which contain integer-valued random left-c.e.\ reals.
The `if' direction is a consequence of \cite{NST}.
For the `only if' direction it suffices to show that every
integer-valued random left-c.e.\ real has high degree.
Let $\alpha$ be an integer-valued random left-c.e.\ real, and $\langle \alpha_s \rangle_{s < \omega}$ 
a computable increasing sequence of rationals converging to $\alpha$. 
We may assume that $\alpha$ has infinitely many 1s, and so by speeding 
up the enumeration we may ensure that $\alpha_s$ has at least $s$ 1's. 
Let Tot $= \{ e\ |\  \varphi_{e} \ \mbox{is total} \}$ be the canonical $\Pi^0_2$ complete set. 
We build a Turing functional $\Gamma$ such that  for all $e$, $\lim_{k} \Gamma^{\alpha}(e,k) = $ Tot$(e)$. 
Then $\emptyset'' \leqslant_{T} \alpha'$ and so $\alpha$ is high. We also construct for each $e \in \omega$ a computable integer-valued martingale $M_e$. 
Let 
\begin{eqnarray*}
d_{e,s}=d_e[s]&=&\max \{ k\ |\  (\forall \sigma \in 2^k) (M_e(\sigma)[s]\downarrow) \}\\
l_{e,s} &=& \max \{ k\ |\  (\forall j < k)(\varphi_e(j)[s]\downarrow) \}. 
\end{eqnarray*}
We proceed in stages $s$, each consisting of two steps.

\ \paragraph{\bf Construction at stage $s+1$}
For each $\langle e,k \rangle \leq s$ do the following
\begin{itemize}
\item[(a)] If $\Gamma^{\alpha}(e,k)[s]\un$, define it as follows. 
If $l_{e,s+1} \geqslant k$ let $\Gamma^{\alpha}(e,k)[s+1] = 1$  with use $\gamma(e,k)[s+1] = 0$; 
otherwise let $\Gamma^{\alpha}(e,k)[s+1] = 0$ 
with use the maximum of $\gamma(e,k)[s]$, $\gamma(e,k-1)[s+1]$, and $h$, where $h$ is the position of the first 1 of $\alpha_s$ after $\alpha_s \upharpoonright d_{e,s}$.

\item[(b)]  If $l_{e,s+1} > l_{e,s}$,
Define $M_e$ to wager 1 dollar on $(\alpha_s \upharpoonright h) \ \hat{\empty} \ 1$, 
and bet neutrally on all other strings with length in $(d_{e,s}, h+1]$, where $h$ is the position of the first 1 of $\alpha_s$ after $\alpha_s \upharpoonright d_{e,s}$.

\end{itemize}

\ \paragraph{\bf Verification}
Since $\alpha$ is a left-c.e.\ real, it follows 
from the construction that $\Gamma$ is well defined, i.e.\ it is consistent.
We show that $\Gamma^{\alpha}$ is total by showing that $\lim_s \gamma(e,k)[s]$ exists for all pairs $(e,k)$, and that $\lim_{k} \Gamma^{\alpha}(e,k) = $ Tot$(e)$. We say that a stage $t$ is \emph{$e$-expansionary} if $l_{e,t} > l_{e,t-1}$. 

First suppose that $e \not\in $ Tot. Then $\lim_s l_{e,s}$ and $\lim_s d_{e,s}$ both exist. Let $\lim_s l_{e,s} = l$ and $\lim_s d_{e,s} = d$, and suppose these limits are reached by stage $s_0$. Let $s_1 \geqslant s_0$ be the least stage where $\alpha_{s_1} \upharpoonright d = \alpha \upharpoonright d$. Then for all $k$ and all stages $s$ where 
$\langle e,k \rangle \geqslant s \geqslant s_1$, 
$\Gamma^\alpha(e,k)[s]$ is set to 0 and $\gamma(e,k)[s]$ is set to be the position of the first 1 after $\alpha_s \upharpoonright d$. As $\alpha$ is left-c.e., the position of the first 1 of $\alpha_s$ after $\alpha_s \upharpoonright d$ at any stage $s \geqslant s_1$ is at most the position of the first 1 of $\alpha_{s_1}$ after $\alpha_{s_1} \upharpoonright d$. Therefore $\lim_s \gamma(e,k)[s]$ exists. For all $k$ such that $\langle e,k \rangle < s_1$, $\lim_s \gamma(e,k)[s]$ is at most $\max_{s < s_1} h_{e,s}$ where $h_{e,s}$ is the position of the first 1 of $\alpha_s$ after $\alpha_s \upharpoonright d_{e,s}$.

Now suppose that $e \in$ Tot. Then there is a sequence of stages $\langle s_i \rangle$ and a sequence $\langle h_i \rangle$ such that we define $M_e$ to wager 1 dollar on $(\alpha_{s_i} \upharpoonright h_i) \ \hat{\empty} \ 1$ at stage $s_i$. Note that $s_i$ is least such that $l_{e,s_i} = i$. 
The real $\alpha$ is left-c.e., so  $\alpha_s \upharpoonright (h_i +1)$ can only move 
lexicographically to the left as $s$ increases.
Moreover, the approximation to $\alpha$ will never extend $(\alpha_{s_i} \upharpoonright h_i) \ \hat{\empty} \ 0$, and so $M_e$ cannot lose capital along $\alpha$. As $\alpha$ is \ivr, $M_e$ does not succeed on $\alpha$. If $\alpha \upharpoonright (h_i + 1) = \alpha_{s_i} \upharpoonright (h_i + 1)$ then $M_e$ increases in capital by 1
dollar.  Therefore there are only finitely many $h_i$ for which $\alpha \upharpoonright (h_i + 1) = \alpha_{s_i} \upharpoonright (h_i +1)$. Let $i_0$ be least such that $\alpha\upharpoonright (h_j + 1) \ne \alpha_{s_j} \upharpoonright (h_j +1)$ for all $j \geqslant i_0$.
We show that $\Gamma^\alpha(e,k) = 1$ for all 
$k \geqslant i_0$, 
thus concluding the proof. 

Suppose by induction that $\Gamma^{\alpha}(e,i) = 1$ for all $i_0 \leqslant i < k$. At stage $s_{k-1}$ we have $l_{e,s_{k-1}} = k-1$. For any stage $t$ with $s_{k-1} \leqslant t < s_k$, if $\Gamma^\alpha(e,k)$ becomes undefined we set $\Gamma^\alpha(e,k)[t] = 0$ and set $\gamma(e,k)[t]$ to be the maximum of $\gamma(e,k)[t-1]$, $\gamma(e,k-1)[t]$, and the position of the first 1 of $\alpha_t$ after $\alpha_t \upharpoonright d_{e,t}$. Let $h$ be the position of the first 1 in $\alpha_{s_k}$ after $\alpha_{s_k} \upharpoonright d_{e,s_{k-1}}$. At stage $s_k$ we see $l_{e,s_k} = k$ and define $M_e$ to wager 1 dollar on $(\alpha_{t'} \upharpoonright h) \ \hat{\empty} \ 1$. 
If $\alpha$ changes below $\gamma(e,k)[s_k-1]$ at stage $s_k$ then $\gamma(e,k)[s_k]$ will be set to at least $h$. Otherwise, $\gamma(e,k)[s_k-1] \geqslant h$. Then as $k \geqslant i_0$, $\alpha$ changes below $h$ at some stage $t' > s_k$. At stage $t'$, $\Gamma^\alpha(e,k)$ will become undefined. At the next $e$-expansionary stage we set $\Gamma^{\alpha}(e,k) = 1$ with use 0.
This concludes the verification and the proof of Theorem \ref{cedegcontivrlceres}.

\subsection{Array computable c.e.\ degrees do not compute \ivrs}\label{subse:ancdegncomivr}
A natural class of c.e.\ degrees that do not contain \ivrs\ is the class of
array computable degrees. In this section we sketch the proof of this fact, which along with
Corollary \ref{coro:anccpivr} gives
Theorem \ref{ancccc} that was presented in the introduction.

By \cite{Ishm, Terwijn.Zambella:01} (also see \cite[Proposition 2.23.12]{roddenisbook}) if $A$ is array computable and c.e., $h$ is
a nondecreasing unbounded function and
$f\leq_T A$, then there exists a computable approximation $(f[s])$ of $f$ such that 
\begin{equation}\label{eq:hca}
| \{ s\ |\ f(x)[s]\neq f(n)[s+1] \} |\leq h(x) \ \ \textrm{for all $x$.}
\end{equation}
Hence given an \ivr\ $B$ and a c.e.\ set $A$ such that $B\leq_T A$, it suffices to define an order function $h$ and a function
$f\leq_T A$ such that any computable approximation $(f[s])$ to $f$ does not satisfy \eqref{eq:hca}.
Let $B$ be \ivr\ and suppose $A$ is c.e. and $\Gamma^A = B$. We assume that at stage $s$, $\Gamma$ has computed $s$ many bits of $\Gamma^A[s]$. We define an order function $h$ and a Turing functional $\Delta$ such that the function $f = \Delta^A$ does
 not satisfy \eqref{eq:hca} for any computable approximation $(f[s])$ of it. Let $\langle \psi_e \rangle$ be an effective list of all binary partial computable functions. We meet the requirements

\begin{description}
\item[$\mathcal{R}_e$] $(\exists x)(f(x) \ne \lim_s \psi_e(x,s) \vee |\{s\ |\ \psi_e(x,s) \ne \psi_e(x,s+1)  \} | \geqslant h(x)). $
\end{description}
\noindent
We define for each $e \in \omega$ an integer-valued martingale $m_e$. 
First, let us describe the strategy
for $\mathcal{R}_0$.
We will have $h(0) = 1$. At stage 1 we define $f(0) = 1$ with use $\delta_1(0) = \gamma_1(0)$. 
We wait until a stage $s$ where we see $\psi_0(0,s) = f(0)[s] = 1$. If this happens, we will want to define $m_0$ to put pressure on $A$ to change so that we may redefine $f(0)$. We define $m_0$ to start with \$1 in capital and wager \$1 on $\Gamma^A[s] \upharpoonright 1$. 
If $\Gamma^A \upharpoonright 1$ changes then we get a change in $A$ below $\gamma_1(0)$, and so a change in $A$ below $\delta_1(0)$. We may therefore redefine $f(0)$ and so meet $\mathcal{R}_0$. 
We assume that there will be no change in $\Gamma^A \upharpoonright 1$, and so we immediately 
look to see whether we can 
start attacking $\mathcal{R}_0$ again by trying to redefine $f(1)$. 
At stage 2 we define $f(1) = 2$ with use $\delta_2(1) = \gamma_2(1)$. 
If we see no change in $\Gamma^A \upharpoonright 1$, then the martingale $m_0$ has \$2 on $\Gamma^A \upharpoonright 1$. 
We will have $h(1) = 1$.
We wait until a stage $s'$ where we see $\psi_0(0,s') = f(0)[s']$ and $\psi_0(1,s') = f(1)[s']$. If this happens, we define $m_0$ to wager \$1 on $\Gamma^A[s'] \upharpoonright 2$. If $A$ changes below $\gamma_{s'}(1)$ then we may redefine $f(1)$ and meet $\mathcal{R}_0$. 

We would like $m_0$ to be total. Therefore whenever we let $m_0$ wager some of its capital on a string $\sigma$, we extend $m_0$ by letting it bet neutrally on all other strings of length at most $|\sigma|$.
Now suppose that none of our previous attempts to redefine $f(0), \ldots,  f(x-1)$ have been successful. 
We wait until a stage $s$ where we have $\psi_e(y,s) = f(y)[s]$ for all $y \leqslant x$. 
The use $\delta_s(x)$ will be equal to $\gamma_s(l)$ for some $l$. 
Suppose we have defined $m_0$ up to strings of length $l-1$ and that $m_0$ has \$$k$ on $\Gamma^A \upharpoonright (l-1)$. 
Suppose $h(x) = n$. Then we require $n$ changes in $A$ to redefine $f(x)$ as many times as we would like. 
If we let $m_0$ wager \$1 on $\Gamma^A \upharpoonright l$ and see $A$ change below $\gamma_s(l)$, we can redefine $f(x)$ once. Suppose that we see this change at stage $t$. We lift $\delta_t(x) = \gamma_t(l+1)$. 
The martingale $m_0$ has been defined up to strings of length $l$, and we have $m_0(\Gamma^A_t \upharpoonright l) = k-1$. 
We again wait until a stage $t'$ where $\psi_e(y,t') = f(y)[t']$ for all $y \leqslant x$. If this occurs, 
we now define $m_0$ to wager \$2 on $\Gamma^A_{t'} \upharpoonright (l+1)$. We do this so that if we do not see a change in this instance, $m_0$'s capital becomes \$$k+1$. 
When we set $\delta_t(x) = \gamma_t(l+1)$ this caused $f(x')$ to become undefined for all $x' > x$. At stage $t' + 1$ we define $f(x+1) = x+2$ with use $\delta_{t'+1}(x+1) = \gamma_{t'+1}(l+2)$. Therefore, if necessary we may start attacking $\mathcal{R}_0$ by trying to redefine $f(x+1)$. 
If every time we see a change for $f(x)$ we increase our wager by \$1, after $n-1$ many changes we are left with \$$k - (1 + 2  + \ldots + n-1) =$ \$$k-\frac{1}{2}(n-1)n$. In attempting to get the $n$th change, we wager all remaining capital and require that if we do not see another change, then we end up with more than \$$k$. So we want $2(k - \frac{1}{2}(n-1)n) > k$. That is, $k > (n-1)n$. We therefore set $h(0) = 1$ and let $h(n)$ be the greatest $m$ such that $(m-1)m < h(n-1) + 1$. If we define $m_0$ as above then either we see all required changes, or $m_0$'s capital increases to at least $k+1$. As $B$ is \ivr, we eventually do see all changes to redefine some $f(x)$, and satisfy $\mathcal{R}_0$.

\subsubsection{Multiple requirements and interactions}
In order to to deal with multiple requirements, we proceed as follows.
The function $f = \Delta^A$ is a global object which must be defined on all inputs. As in the strategy above, the values $f(x)$ are changed by the action of the requirements. Suppose we satisfy $\mathcal{R}_0$ by redefining $f(0)$ once. We could attempt to satisfy $\mathcal{R}_1$ by further redefining $f(0)$, but at some point we must stop. We choose a fresh large number $x_1$, and have the strategy for $\mathcal{R}_1$ try to redefine $f(x_1)$ as many times as necessary. As we saw above, the strategy for $\mathcal{R}_1$ may at any one time be wanting 
to redefine $f(y)$ for possibly many $y$. We formalise this by associating to each requirement $\mathcal{R}_e$ at stage $s$ an interval $I_{e,s}$ of natural numbers, so that $\mathcal{R}_e$ at stage $s$ is wanting to redefine $f(x)$ for $x \in I_{e,s}$. 
When we are successful in redefining $x \in I_{e,s}$, we remove all $y > x$ from $I_{e,s}$. If we have not already satisfied $\mathcal{R}_e$ at some later stage $s'$ and we see $\psi_e(z,s') = f(z)[s']$ for all $z \leqslant x+1$, then we add $x+1$ to $I_{e,s'}$ and attempt to redefine $f(x+1)$ as well. 

Consider the requirements $\mathcal{R}_e$ and $\mathcal{R}_f$, with $\mathcal{R}_e$ of stronger priority than $\mathcal{R}_f$. We are defining martingales $m_e$ for $\mathcal{R}_e$ and $m_f$ for $\mathcal{R}_f$. It is possible that when $\Gamma^A$ moves and we redefine some $f(k)$ for the sake of $\mathcal{R}_e$ that the martingale $m_f$ also loses capital, even though we do not redefine some $f(j)$ for the sake of $\mathcal{R}_f$. We will therefore want to start a new version of $m_f$ every time a requirement of stronger priority than $\mathcal{R}_f$ acts. 
We say that $\mathcal{R}_e$ requires attention at stage $s$ if one of the following holds:

\begin{enumerate}
\item $I_{e,s} = \emptyset$.
\item for all $x \leqslant \max I_{e,s}$ we have $\psi_e(x,s) = f(x)[s]$ and $|\{ t < s : f(x)[t] \ne f(x)[t+1]  \} | < h(x)$, and 
$A_s(z) \ne A_{s-1}(z)$ for some $z \in (\delta_s(\min I_{e,s}-1),\delta_s(\max I_{e,s})]$.
\item for all $x \leqslant \max I_{e,s}$ we have $\psi_e(x,s) = f(x)[s]$ and $|\{ t < s : f(x)[t] \ne f(x)[t+1]  \} | < h(x)$, and 
$\psi_e(\max I_{e,s} +1,s) = f(\max I_{e,s} +1)[s]$.
\end{enumerate}
We are ready to produce the construction.

\subsubsection{Construction}
At stage 0,
define $m_e(\lambda) = 1$ for all $e \in \omega$. 
Let $f(x)[0] = \Delta^A(x)[0] = 1$ with use $\delta_0(x) = x$ for all $x \in \omega$.
Let $I_{e,1} = \emptyset$ for all $e \in \omega$. 
Each stage of the construction
 after stage 0 
consists of three steps.
At stage $s, s \geqslant 1$ proceed as follows: 

\textit{Step 1}:
For all $e \leqslant s$, if a requirement of stronger priority than $\mathcal{R}_e$ has acted since $\mathcal{R}_e$ last acted, we start a new version of $m_e$, and define $m_e(\lambda) = 1$. Otherwise, we continue with the previous version of $m_e$. Let $d_{e,s}$ denote the length of the longest string for which the current version of $m_e$ is defined. 

\textit{Step 2}: 
Let $x$ be least such that $f(x)$ is undefined at the beginning of stage $s$. (If there is no such $x$, proceed to the next step.) 
Let $l = \max_{e \leqslant s} d_{e,s}$. Define $f(x)[s] = s+1$ with use $\delta_s(x) = \gamma_s(l+1)$. 

\textit{Step 3}:
Let $\mathcal{R}_e$ be the requirement of strongest priority which requires attention at stage $s$. 
Choose the first case by which $\mathcal{R}_e$ requires attention.

If case 1 holds, choose a fresh large number $x_e$ and let $I_{e,s+1} = \{ x_e \}$. 

If case 2 holds, then let $x \in I_{e,s}$ be least such that $A_s(z) \ne A_{s-1}(z)$ for some $z \in (\delta_s(x-1),\delta_s(x)]$. 
Let $f(x) = s+1$ with use $\delta_{s+1}(x) = \gamma_{s+1}(d_{e,s}+1)$. 
We have that $m_e(\Gamma^A[s] \upharpoonright d_{e,s}) < m_e(\Gamma^A[s] \upharpoonright (d_{e,s} - 1))$. 
If $m_e(\Gamma^A[s] \upharpoonright d_{e,s}) \ne 0$, let $n_e = \max_{i \leqslant d_{e,s}} h(m_e(\Gamma^A[s] \upharpoonright i)$). 
Suppose that $j$ is such that $n_e = h(m_e(\Gamma^A[s] \upharpoonright j)$ and let $n'_e = | \{ m_e(\Gamma^A[s] \upharpoonright i) : j \leqslant i \leqslant d_{e,s} \} |$. Then we have received $n'_e$ of the $n_e$ permissions required to redefine $f(x)$ at least $h(x)$ many times. If $n'_e = n_e - 1$, then define $m_e$ to wager $\Gamma^A[s] \upharpoonright d_{e,s}$ dollars on $\Gamma^A[s] \upharpoonright (d_{e,s} + 1)$. Otherwise let $w = m_e(\Gamma^A[s] \upharpoonright (d_{e,s}-1)) - m_e(\Gamma^A[s] \upharpoonright d_{e,s})$ and define $m_e$ to wager \$$(w+1)$ on $\Gamma^A[s] \upharpoonright (d_{e,s} + 1)$. 
If $m_e(\Gamma^A[s] \upharpoonright d_{e,s}) = 0$, let $m_e$ bet neutrally on all other strings of length $d_{e,s} + 1$. Let $I_{e,s+1} = [ \min I_{e,s}, x]$.

If case 3 holds, then let $I_{e,s+1} = I_{e,s} \cup \{ \max I_{e,s} + 1 \}$. Define $m_e$ to wager \$1 on $\Gamma^A[s] \upharpoonright (d_{e,s} + 1)$.

In any case, let $I_{e',s+1} = \emptyset$ for all $e' > e$. 
\subsubsection{Verification}
We need to show that
for all $e \in \omega$, $\mathcal{R}_e$ is satisfied.
Assume by induction that stage $s^*$ is the last stage at which a requirement of stronger priority than $\mathcal{R}_e$ acts.
Assume for all $s \geqslant s^*$ that $\psi_e(x,s) = f(x,s)$ for all $x \leqslant \max I_{e,s}$. 
At stage $s^*+1$ we will define a new version of $m_e$, which will be the final version. 
At every stage after $s^*+1$, we define more of $m_e$. Therefore $m_e$ is total. As $\Gamma^A$ is \ivr, $m_e(\Gamma^A) = \sup \, \{ m_e(\Gamma^A \upharpoonright i) : i \in \omega \} < \infty$. Let $ \sup \{ m_e(\Gamma^A \upharpoonright i) : i \in \omega \} = k$ and $i_0$ be such that $m_e(\Gamma^A \upharpoonright i_0) = k$. Suppose $s_0$ is least such that $s_0 \geqslant s^*+1$ and $\Gamma^A[s_0] \upharpoonright i_0 = \Gamma^A \upharpoonright i_0$, and $x$ is such that $\delta_{s_0}(x) = \gamma_{s_0}(d_{e,s_0}+1)$. Then $f(x)$ is redefined $h(x)$ many times and $\mathcal{R}_e$ is satisfied.

\section{Jump inversion for \ivrs}\label{se:jumpinvivrs}
Jump inversion for \ml randoms was discovered in \cite{MR820784, MR2170569} and
was generalized in \cite{BDNGP}. Every
degree which is c.e.\ in and above $\mathbf{0}'$ contains the jump of some \ml random
$\Delta^0_2$ set. Hence the same holds for the \ivrs. However in this case we can obtain a stronger
jump inversion theorem by requiring that the `inverted' degrees are c.e. Note
that this stronger theorem does not hold for \ml randoms since 
$\mathbf{0}'$ is the only c.e.\ degree containing a \ml random. Moreover it does not hold
for computable randomness or Schnorr randomness, since by  \cite{NST} the only c.e.\ degrees
that contain such randoms are high. integer-valued randomness is the strongest known randomness notion
for which jump inversion holds with c.e.\ degrees.

Since the argument is somewhat involved, we present it in two steps.
In Section \ref{subse:lowcecontivr} we discuss the strategy for controlling the jump of an \ivr \ of c.e.\ degree.
This argument gives a low c.e.\ degree which contains an \ivr. 
It is a finite injury construction, and the hardest of the two steps.
Our argument actually shows the stronger result that there is a low c.e.\ weak
truth table degree which contains an \ivr.
We can then add coding requirements in order to prove the full jump inversion theorem, which we present in full detail in
Section  \ref{subse:fulljumpcecontivr}. This construction is a tree argument which uses
the strategies of Section \ref{subse:lowcecontivr} for ensuring that the jump of the constructed set is below the given
$\Sigma^0_2$ set, combined with standard coding requirements which deal with the remaining requirements.

\subsection{A low c.e.\ degree containing an \ivr}\label{subse:lowcecontivr}
We build an integer-valued random $A$ of low c.e.\ degree.
In fact, we build an integer-valued random $A$ and a c.e.\ set $B$ such that  $A \equiv_{wtt} B$.
Let $\langle m_e \rangle$ be an effective list of all partial integer-valued martingales. 
In order to ensure that $A$ is \ivr\ it suffices to satisfy the 
following requirements:

\begin{description}
	\item[$R_e$\ ] \hspace{0.3cm}if $m_e$ is total, then $m_e$ does not succeed on $A$;
	\item[$N_e$\ ] \hspace{0.3cm}$(\exists^{\infty}s)\ (\Phi_e^A(e)[s] \downarrow) 
	\implies \Phi_e^A(e)\downarrow$.
\end{description}
\noindent
We order the requirements as $R_0 > N_0 > R_1 > N_1 > \ldots$ and 
begin by setting $A_1 = 1^{\omega}$. To meet $R_0$, we observe the values of the martingale $m_0$. If $m_0$ increases its capital along $A$, we change $A$ to force $m_0$ to lose capital. 
As $m_0$ is integer-valued, if it loses capital, it must lose at least \$1. Thus if we can force $m_0$ to lose capital every time we act, we need only act for $R_0$ finitely many times. 
As we are building reductions $\Gamma$ and $\Delta$ such that $\Gamma^B = A$ and $\Delta^A = B$, to change $A$ we will need to change $B$. Once we have changed $B$, we will then need to change $A$ again to record this fact.
To satisfy the requirement $N_e$ we use the usual strategy of preserving the restraint $\varphi_e^A(e)[s]$ at all but finitely many stages. As the strategy for an $R$-requirement is finitary, this can be done easily.

\subsubsection{The finite injury construction}\label{subsubse:basfininconslev}
In order to help with the definition of the reductions, we make use of 
\textit{levels} $\langle l_i \rangle_{i < \omega}$ and $\langle d_i \rangle_{i < \omega}$. 
We calculate the size of the levels below. We set $\gamma(l_i) = d_i$ and $\delta(d_i) = l_{i+1}$. We say that we act for requirement $R_e$ {\em at level} $l_{i+1}$ at stage $s$ if we change $A$ to decrease $m_e$'s capital from $A \upharpoonright l_i$ to $A \upharpoonright l_{i+1}$. That is,  $m_e(A_s \upharpoonright l_{i+1}) > m_e(A_s \upharpoonright l_i)$ and $m_e(A_{s+1} \upharpoonright l_{i+1}) < m_e(A_{s+1} \upharpoonright l_i)$. We act at level $l_{i+1}$ only for the sake of the requirements $R_0, \ldots, R_i$. Once we have acted at level $l_{i+1}$, we enumerate an element from $[d_i, d_{i+1})$ into $B$. 
To record this change in $B$, we let $A$ extend a string of length $l_{i+2}$ which has not yet been visited.
So that the reduction $\Delta$ is consistent, we must not let $A$ extend a string which is {\em forbidden}, that is, a string $\sigma$ such that $\Delta^\sigma \not\prec B_{s+1}$. We carefully define the levels $l_i$ so that the action from the requirements never forces us to extend a forbidden string.


Before we define $\langle l_i \rangle_{i < \omega}$ and $\langle d_i \rangle_{i < \omega}$,
we lay out the construction (in terms of these unspecified parameters and a function $d$ defined below). In Section \ref{subsubse:towardeflev}
we discuss the  various properties that these parameters need to satisfy in order that the construction
is successful (i.e.\ produces sets $A, B$ which satisfy the requirements $R_e, N_e$).


We have for every $e \in \omega$ and every stage $s$ a 
restraint $r_{e,s}$. We say that $R_e$ requires attention at level $l_{i+1}$ at stage $s$ if

\begin{enumerate}
\item $m_e(\sigma)[s] \downarrow$ for all strings $\sigma$ of length $\leq l_{i+1}$,
\item $l_{i+1} \geqslant r_{e,s}$
\item $m_e(A_s \upharpoonright l_{i+1}) > m_e(A_s \upharpoonright l_i)$.
\end{enumerate}
\noindent
We say that $R_e$ requires attention at stage $s$ if it does so at some level. 
We say that $N_e$ requires attention at stage $s$ if $\Phi_e^A(e)[s] \downarrow$.
Recall the definition of the {\em sibling} of a string, just after 
Definition \ref{de:intvalmar}.

\ \paragraph{{\em Construction}}
Let $\gamma(l_i) = d_i$ and $\delta(d_i) = l_{i+1}$.
At stage 0, let $A_1 = 1^{\omega}$ and $r_{e,1} = l_{e}$ for all $e$.

\noindent \textit{Stage $s, s \geqslant1$}: Find the requirement of 
strongest priority which requires attention at stage $s$. (If there is no such requirement, proceed to the next stage.) 

\textit{Case 1}: If this is $R_e$, let $l_{i+1}$ be least such that $R_e$ requires attention at level $l_{i+1}$ at stage $s$. 
Let $l \in (l_i, l_{i+1}]$ be least such that $m_e(A_s \upharpoonright l) > m_e(A_s \upharpoonright l_i)$. 
Choose a string $\tau$ of length $l_{i+1}$ which extends the sibling of $A_s \upharpoonright (l-1)$ such that 
the minimum of all $d(\tau,\mu)$, where $\mu$ is any forbidden 
string of length $l_{i+1}$ extending $A_s \upharpoonright l_i$, is as large as possible.
Enumerate an element of $[d_i,d_{i+1})$ into $B$. 
Choose a string $\rho$ of length $l_{i+2}$ extending $\tau$ such that 
$\rho \not\prec A_t$ for all $t < s$, and 
the minimum of all $d(\rho,\mu)$, where the minimum is taken over forbidden strings $\mu$ of length $l_{i+2}$ extending $\tau$, 
is as large as possible.
Set $A_{s+1} = \rho1^{\omega}$. For all $e' \geqslant e$ with $r_{e',s} \leqslant l_{i+1}$, let $r_{e',s+1} = l_{i+1}$.

\textit{Case 2}: If this is $N_e$, for all $e' > e$ with $r_{e',s} \leqslant \varphi_e^A(e)[s]$, let $r_{e',s+1} = \varphi_e^A(e)[s]$.

In the following section we give the remaining specifications and analysis of the construction, as well as
the verification.

\subsubsection{The calculation of the levels \texorpdfstring{$l_i$}{l_i} and
\texorpdfstring{$d_i$}{d_i} for a successful construction}\label{subsubse:towardeflev}
In the following we calculate the levels
 $l_i$, $d_i$, and depict this process in Figure \ref{fig:calclevelli}.
Suppose that we act at level $l_{i+1}$ for $R_e$ and naively let $A$ extend a string $\tau$ of length
 $l_{i+1}$ whose sibling is forbidden. Consider the situation where $m_0$ increases its capital on the 
 very last bit of $A_s \upharpoonright l_{i+1}$, loses capital on $\tau$, and is neutral on all other strings of length
 $l_{i+1}$. We will not be able to change $A$ to extend $\tau$'s sibling, as this string is forbidden. However, 
 we do not want to change $A$ so that $m_0$ is neutral, as we would like the action for $R_0$ to be finitary.
To avoid such a situation we must be more careful in how we change $A$. In particular, we must 
ensure that $A$ is kept in some sense ``far away'' from forbidden strings. This is made precise below. 

We first calculate an upper bound on the number of forbidden strings of length $l_{i+1}$ 
which can occur above a nonforbidden string of length $l_{i}$. Our upper bound will not be strict. 
A string $\sigma$ of length $l_{i+1}$ becomes forbidden if $\Delta^\sigma$ is no longer giving correct $B$-information. As $\delta(d_i) = l_{i+1}$, $\Delta^\sigma$ will be incorrect 
only if we enumerate an element into $B$ below $d_i$, which occurs only when we act for a
 requirement at some level $\leqslant l_i$.  We will act for $R_e$ at level $l_{i+1}$ only when we
  see $m_e$ halt on all strings of length $l_{i+1}$, and so if $A$ no longer changes below $l_i$, we will act for $R_e$ at level $l_{i+1}$ at most once.
As we act at level $l_{i+1}$ only for the sake of requirements $R_0, \ldots, R_i$, if $A$ no longer 
changes below $l_i$, we can act at level $l_{i+1}$ at most $i+1$ times. After acting at a level $l_j$ 
for some $j \leqslant i$, we allow $R_0, \ldots, R_{i}$ to act at level $l_{i+1}$ again. We begin with
 $A_1 = 1^{\omega}$. Suppose we act $i+1$ times at level $l_{i+1}$. We then act at level $l_i$. 
 We act another $i+1$ times at level $l_{i+1}$ before we again act at level $l_i$. We can act at level 
 $l_i$ at most $i$ times. This can continue until we get to level $l_1$, where we can change $A$ once 
 below $l_1$ for the sake of requirement $\mathcal{R}_0$. Therefore we act at a level $\leqslant l_i$ $2.3.4. \ldots \cdot (i+2) = (i+2)!$ many times, and there are at most 
\[
\parbox{4cm}{$f_{i+1} = \sum_{j=0}^i \, (j+1)!$}
\]
many forbidden strings of length $l_{i+1}$.
Note that for any $k \in \omega$ we may enumerate all partial integer-valued martingales with initial capital $k$. 
We therefore may assume that our list $\langle m_e \rangle$ of all partial integer-valued martingales comes with
 a computable intial capital, $m_e(\lambda)$. As a martingale may at most double its capital in a single bet, the 
 upper bound on $m_e$'s capital at a string of length $n$ is $2^nm_e(\lambda)$. 

We now show how a martingale can force us ``closer'' to a forbidden string. Suppose at stage $s$ that $A_s$ 
extends the string $\nu$ of length $l_{i+1}$, and there is a forbidden string $\mu$ of length $l_{i+2}$ above. 
For simplicity, suppose that $A_s$ extends the leftmost string of length $l_{i+2}$ which extends $\nu$, and that $\mu$ is the rightmost string of length $l_{i+2}$ which extends $\nu$.
If $\mathcal{R}_j$ requires attention at level $l_{i+2}$, we would like to choose a string $\tau$ of length $l_{i+2}$ with $m_j(\tau) < m_j(\nu)$. 
The problem is the following. Suppose that $m_j$ increases its capital on all string of length $l_{i+2}$ 
which extend $\nu0$. By Kolmogorov's inequality, this must mean that $m_j(\nu0) > m_j(\nu)$, and so 
$m_j(\nu1) < m_j(\nu)$. If $m_j$ has sufficient capital at $\nu1$, it may then increase its capital above 
$m_j(\nu)$ on all strings of length $l_{i+2}$ which extend $\nu10$. Again by Kolmogorov's inequality 
we have $m_j(\nu10) > m_j(\nu1)$ and so $m_j(\nu11) < m_j(\nu1) < m_j(\nu)$. Now $m_j$ is an 
integer-valued martingale, and so after doing this finitely many times, say $n$ times, we have
 $m_j(\nu1^n) < \frac{1}{2}m_j(\nu)$ and so $m_j$ cannot increase its capital above $m_j(\nu)$
  on all strings of length $l_{i+2}$ which extend $\nu1^n0$. If $l_{i+2} \geqslant l_{i+1} + n + 1$, 
  then we can pick a string $\tau$ extending $\nu1^n0$ with $m_j(\tau) < m_j(\nu)$ and which is 
  not forbidden. For two strings $\alpha$ and $\beta$ of length $l$, let $d(\alpha,\beta)$ be $l-b$, 
  where $b$ is the length of the longest common initial segment. Then in this situation, we have 
  $d(A_s,\mu) = l_{i+2} - l_{i+1}$ and $d(\tau,\mu) \leqslant l_{i+2} - l_{i+1} - n$. Therefore $m_j$ 
  has forced $A$ distance $n$ closer to the forbidden string $\mu$. Now $\mathcal{R}_j$ might not 
  be the only requirement which can act at level $l_{i+2}$. We will then need to calculate the distance 
  that the other martingales may move $A$, and ensure that $l_{i+2}$ is high enough. 

We calculate a bound on how far an integer-valued martingale $m$ may move $A$. If $m(\nu) = k$ 
and $m$ increases its capital to $k+1$ on all strings extending $\nu0$, then $m(\nu1) = k-1$. If 
$m(\nu1) \geqslant 2$ then $m$ can increase its capital to $k+1$ on all strings extending $\nu10$. 
Then $m(\nu11) = k - 1 - 2$. If $m(\nu1) \geqslant 4$ then $m$ can increase its capital to $k+1$ on
 all strings extending $\nu110$. Then $m(\nu111) = k - 1 - 2 - 4$. When $m(\nu1^n) = k - 1 - 2 - \ldots - 2^{n-1} < \frac{1}{2}k$,
  $m$ is not able to increase its capital to $k+1$ on all strings extending $\nu1^n0$. We let 
  $n(k) = (\mu n)(k - 1 - 2 - \ldots - 2^{n-1} < \frac{1}{2}k)$. Then $m$ can move $A$ at most a 
  distance $n(k)$. In the case where $A_s$ is not the leftmost string extending $\nu$ and $\mu$ 
  is not the rightmost string extending $\nu$, a similar argument shows that $m_e$ can still move 
  $A$ a distance of at most $n(k)$. The only difference is that $m_e$ would then need to bet against 
  the initial segments of $\mu$ which are of length greater than $l_{i+2}-l_{i+1}-d(A_s,\mu)$. 

Suppose we act at level $l_{i+1}$ at stage $s$ and let $A$ extend the string $\nu$ of length $l_{i+1}$
 which has forbidden strings of length $l_{i+2}$ above. We enumerate an element of $[d_i,d_{i+1})$
  into $B$. To record this change in $B$, we choose a string $\rho$ of length $l_{i+2}$ which has not 
  been visited before, and which is as far from any forbidden string as possible. We have that there are 
  at most $f_{i+2}$ many forbidden strings of length $l_{i+2}$ above a nonforbidden string of length $l_{i+1}$. 
Let $x = (\mu x)(2^x \geqslant f_{i+2})$. By the counting argument for $f_{i+1}$ given above, if 
$l_{i+2} - l_{i+1} = h > x$, then there is a string of length $l_{i+2}$ which has not been visited yet,
 and which is at least distance $h-x$ from a forbidden string. 
Now suppose that $\mathcal{R}_j$ requires attention at level $l_{i+2}$. We know that there are no
 forbidden strings above $\rho \upharpoonright l_{i+1} + x + 1$, and so if we can reduce $m_j$ by
  moving to a string which is still above $\rho \upharpoonright l_{i+1} + x + 1$, we will do so. Otherwise,
   $m_j$ will move us closer to a forbidden string. The bound on the capital of $m_j$ at 
   $A_s \upharpoonright (l_{i+1} + x + 1)$ is $2^{l_{i+1} + x + 1}m_j(\lambda)$. So we know that
    $m_j$ may move us a distance at most $n(2^{l_{i+1} + x + 1}m_j(\lambda))$ towards a forbidden string.
     If $l_{i+2} - l_{i+1} > x + n(2^{l_{i+1} + x + 1}m_j(\lambda))$ then we will be able to choose a nonforbidden string $\rho'$ which decreases $m_j$. Suppose that $\mathcal{R}_k$, which is of stronger priority than $\mathcal{R}_j$, requires attention at level $l_{i+2}$ at stage $s'$. 
Our reasoning is similar to the previous case. Let $n_0 = n(2^{l_{i+1} + x + 1}m_j(\lambda))$. We know that there are no forbidden strings above $\rho' \upharpoonright l_{i+1} + x + n_0 + 1$, and so if we can reduce $m_k$ by moving to a string which is still above $\rho' \upharpoonright l_{i+1} + x +  n_0 + 1$, we will do so. Otherwise, $m_k$ will move us closer to a forbidden string. The bound on the capital of $m_k$ at $A_{s'} \upharpoonright (l_{i+1} + x + n_0 +1)$ is $2^{l_{i+1} + x + n_0 + 1}m_k(\lambda)$. So we know that $m_k$ may move us a distance at most $n(2^{l_{i+1} + x + n_0 + 1}m_k(\lambda))$ towards a forbidden string. If $l_{i+2} - l_{i+1} > x + n_0 + n(2^{l_{i+1} + x + n_0 + 1}m_k(\lambda))$ then we will be able to choose a nonforbidden string $\rho''$ which decreases $m_k$. We will need $l_{i+2}$ to be large enough so that we can always move in this way for any requirement which might act at level $l_{i+2}$. 

The requirements $\mathcal{R}_0, \ldots, \mathcal{R}_i$ can act at level $l_{i+1}$. We do not know the order in which the requirements may act, so we will have to take the maximum of the capitals of the martingales $m_0, \ldots, m_i$ in our calculation. We illustrate this definition in Figure \ref{fig:calclevelli}
.
We set $l_0 = 0$. Given $l_i$, we set 
$l_{i+1,0} = l_i + (\mu x)(2^x \geqslant f_{i+1})$ and for $0 \leqslant j \leqslant i$, 
\[ 
l_{i+1,j+1} = l_{i+1,j} + \max_{k \leqslant i} \ n(2^{l_{i+1,j} + 1}m_k(\lambda)) 
\]
\noindent 
and let $l_{i+1} = l_{i+1,i+1} + 1$.
The levels $d_i$ are chosen so that we can enumerate an element into $[d_i,d_{i+1})$ every time 
we act at level $l_{i+1}$. This calculation is the same as that of $f_{i+1}$. Let $d_0 = 0$. Given $d_i$, let $d_{i+1} = d_i + \sum_{j=0}^i (j+1)!$.

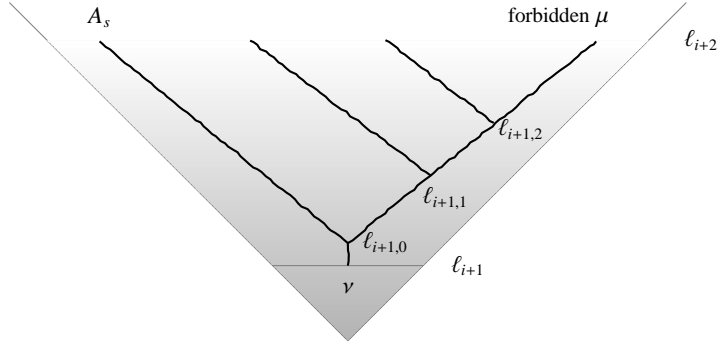
\begin{figure}
\begin{tikzpicture}
\draw [gray] (2,0) -- (-2.5,4.5); 
\draw [gray] (2,0) -- (6.5,4.5);
\fill[bottom color = black!30!white, top color = white] (2,0) -- (-2,4) -- (6,4);
\draw [gray] (1,1) -- (3,1);
\draw [decorate, decoration={random steps,segment length=2pt,amplitude=0.4pt}, thick] (2,1) -- (2,1.3); 
\draw [decorate, decoration={random steps,segment length=2pt,amplitude=0.4pt}, thick] (2,1.3) -- (-1.3,4); 
\draw [decorate, decoration={random steps,segment length=2pt,amplitude=0.4pt}, thick] (3.1,2.2) -- (0.7,4); 
\draw [decorate, decoration={random steps,segment length=2pt,amplitude=0.4pt}, thick] (3.94,2.9) -- (2.5,4); 
\draw [decorate, decoration={random steps,segment length=2pt,amplitude=0.4pt}, thick] (2,1.3) -- (5.3,4); 
\node at (2,0.7) {\small $\nu$};
\node at (3.6,1.0) {\small $\ell_{i+1}$};
\node at (6.7,4) {\small $\ell_{i+2}$};
\node at (-1.3,4.3) {\small $A_s$};
\node at (2.5,1.3) {\small $\ell_{i+1,0}$};
\node at (3.3,1.9) {\small $\ell_{i+1,1}$};
\node at (4.3,2.8) {\small $\ell_{i+1,2}$};
\node at (4.8,4.3) {\small {\footnotesize forbidden} $\mu$};
\end{tikzpicture}
\caption{Calculating the levels $l_i$ and avoiding the forbidden strings.}
\label{fig:calclevelli}
\end{figure}
\subsubsection{Verification of the finite injury construction}
First, we show that $A \equiv_{wtt} B$.
As in the calculation of $f_{i+1}$, we can act at level $l_{i+1}$ at most $\sum_{j=0}^i (j+1)!$ many times. 
We have that $d_{i+1} = d_i + \sum_{j=0}^i (j+1)!$ for all $i$. 
Every time we act at level $l_{i+1}$ and change $A$ below $l_{i+1}$, we enumerate an element from $[d_i,d_{i+1})$ into $B$. 
The uses $\gamma(l_i)$ are clearly computable, and so we have $\Gamma^B = A$ via the weak truth-table functional $\Gamma$.
For the other reduction, 
note that the consistency of $\Delta$ is 
a consequence of $A$ never extending a 
forbidden string. Again the uses $\delta(d_i)$ are computable 
and so $\Delta^A = B$ via the weak truth-table functional $\Delta$. 

Next we argue that all $N_e$ requirements are met. Suppose inductively that all requirements of stronger priority than $N_e$ do not act
after a certain stage $t$. If at some stage $s$ after stage $t$ the computation in requirement $N_e$ halts, then a restraint $r_{e,s}$ will be erected so that
the use of the computation is protected from further enumerations into $A$. Therefore in that case the computation actually halts. Therefore $N_e$ is met, and this concludes
the induction step.

It remains to show that for every $e \in \omega$, $\mathcal{R}_e$ is satisfied.
Suppose by induction that all requirements of stronger priority than $\mathcal{R}_e$ do not act after stage $s^*$. Let $i_0$ be least such that $l_{i_0} \geqslant r_{e,s^*}$. We show that $m_e(A) \leqslant m_e(A_{s^*} \upharpoonright l_{i_0})$. 
Suppose at stage $s \geqslant s^*$ we see $m_e$ increase its capital beyond $m_e(A_{s^*} \upharpoonright l_{i_0})$. Then $\mathcal{R}_e$ will require attention at stage $s$. Suppose that $\mathcal{R}_e$ requires attention at level $l_{i+1}$ at stage $s$. 
Let $t < s$ be the last stage at which we acted for some requirement at a level below $l_{i+1}$. Suppose we acted at level $l_{j+1}$. At stage $t$ we chose some string $\rho$ of length $l_{j+2}$ and let $A_t = \rho1^{\omega}$. Then $A_{s-1} \upharpoonright l_i = A_t \upharpoonright l_i$. If $j = i-1$ then we chose $\rho$ which would have been at least distance $l_{i+1} - l_i - x$, where $x = (\mu x)(2^x \geqslant f_{i+1})$, from any forbidden string of length $l_{i+1}$ (that is, $\rho$ and any forbidden string have a common initial segment of length at most $l_i + x$). Otherwise $j < i - 1$ and there is no forbidden string of length $l_{i+1}$ above $A_s \upharpoonright l_i$. Suppose that between stages $t$ and $s-1$ inclusive we acted at level $l_{i+1}$ $k$ many times. We have that $k < i+1$. 
Then $A_s \upharpoonright l_i$ and any forbidden string have a common initial segment of length at most $l_{i+1,k}$. Let $l \in (l_i,l_{i+1}]$ be least such that $m_e(A_s \upharpoonright l) > m_e(A_s \upharpoonright l_i)$. If $l > l_{i+1,k} + 1$, then there is a string $\tau$ above $A_s \upharpoonright l_{i+1,k} + 1$ with $m_e(\tau) < m_e(A_s \upharpoonright l)$ and which is not forbidden. Otherwise $m_e$ can move us at most distance $n(2^{l_{i+1,k}+1}m_e(\lambda))$ closer to a forbidden string. We have that $l_{i+1} > l_{i+1,k} + n(2^{l_{i+1,k}+1}m_e(\lambda))$, so we can find a nonforbidden string $\tau$ with $m_e(\tau) < m_e(A_s \upharpoonright l)$. Restraints are then imposed so that $\mathcal{R}_e$ and no other requirement of weaker priority may act at level $l_{i+1}$ after stage $s$. Therefore $A \upharpoonright l_{i+1} = A_{s+1} \upharpoonright l_{i+1}$ and so $m_e(A) < m_e(A_{s^*} \upharpoonright l_{i_0})$. 

\subsection{The full jump inversion theorem for \ivrs}\label{subse:fulljumpcecontivr}
Given  a set $S\geq_T \mathbf{0}'$ which is c.e.\ in $\mathbf{0}'$ we show how to construct 
an integer-valued random set $A$ of c.e.\ degree such that $A' \equiv_{T} S$. 
Along with $A$, we build a c.e.\ set $B$ such that $A \equiv_T B$, and show that $B' \equiv_T S$.
Let $\langle m_e \rangle$ be an effective enumeration of all 
partial computable integer-valued martingales. So that $A$ is integer-valued random, we meet the requirements

\begin{description}
	\item[$\mathcal{R}_e$] If $m_e$ is total, then $m_e$ does not succeed on $A$.
\end{description}

\noindent 
We also build wtt-reductions $\Gamma$ and $\Delta$ such that $\Gamma^B = A$ and $\Delta^A = B$. 
For the requirement 
$S \leqslant_T B'$, we build a functional $\Lambda$ and meet the requirements 

\begin{description}
	\item[$\mathcal{P}_e$] $\lim_t \Lambda^B(e,t) = S(e)$. 
\end{description}

\noindent The basic strategy for a $\mathcal{P}$-requirement is as follows. As $S$ is c.e.\ in and above $\mathbf{0}'$, 
we know that $S$ is $\Sigma^0_2$. Therefore there is some computable approximation $\{ S_i \}_{i \in \omega}$ such that $n \in S$ if and only if there is an $s$ such that $n \in S_t$ for all $t > s$. We define $\Lambda^B(e,s) = 1$ for larger and larger $s$ with some large use $\lambda(e,s)$. If we see at some stage $u$ that $e \not\in S_u$ and $u > t$, then we enumerate $\lambda(e,t)$ into $B$ and redefine $\Lambda^B(e,t) = 0$ with use $-1$, i.e. the axiom defining $\Lambda^B(e,t) = 0$ does not depend on $B$. 

For the requirement 
$B' \leqslant_T S$ we \textit{attempt} to meet the requirements

\begin{description}
	\item[$\mathcal{N}_e$] $(\exists^{\infty}s)(\Phi_e^B(e)\downarrow) \implies \Phi_e^B(e)\downarrow$.
\end{description}

\noindent We attempt to meet these as usual by restraining $B$ below the use $\varphi_e^B(e)[s]$ whenever we see $\Phi_e^B(e)[s]\downarrow$. Although we will not actually meet these requirements (doing so would mean that $B' \equiv_T \emptyset'$), trying to meet the requirements will allow us to show that $B' \leqslant_T S$.

\subsubsection{The priority tree}
The construction will use a tree of strategies. 
To define the tree, we specify recursively the association of nodes to requirements, and specify the outcomes of nodes working for particular requirements. To specify the priority ordering of nodes, we specify the ordering between outcomes of any node.
We order the requirements as 
\[
\mathcal{R}_0 > \mathcal{P}_0 > \mathcal{N}_0 > \mathcal{R}_1 > \mathcal{P}_1 > \mathcal{N}_1 > \cdots
\]
and specify that all nodes of length $k$ work for the $k^{th}$ requirement on the list. 
We will have nodes dedicated to $\mathcal{R}$-, $\mathcal{P}$-, and $\mathcal{N}$-requirements. A node dedicated to a $\mathcal{P}$-requirement will have the $\infty$ outcome, corresponding to enumerating infinitely many markers $\lambda(e,s)$, and the $f$ outcome, for when only finitely many markers are enumerated. Suppose that the node $\alpha$ works for $\mathcal{P}_e$ and $\beta$ works for $\mathcal{N}_f$ with $\alpha \prec \beta$. If $f$ is the true outcome of $\alpha$ and $\alpha \ \hat{\empty} \ f \preccurlyeq \beta$, then only finitely many markers are enumerated, and $\beta$ does not need to worry about the computations it sees being destroyed infinitely many times. Now suppose that $\infty$ is the true outcome of $\alpha$ and $\alpha \ \hat{\empty} \ \infty \preccurlyeq \beta$. Then $\beta$ will be guessing that $\mathcal{P}_e$ will enumerate all its unrestrained markers $\lambda(e,s)$ into $B$. It will then not believe a computation $\Phi_e^B(e)[s]$ until it sees that all unrestrained markers below the use $\varphi_e^B(e)[s]$ have been enumerated. This is formalized with the definition of a $\beta$-correct computation below. The outcomes of $\mathcal{R}$- and $\mathcal{N}$-nodes are $\ldots < n < \ldots < 1 < 0$, corresponding to the restraint they impose on $B$.

\subsubsection{Making the sets \texorpdfstring{$A,B$}{A,B} of the same degree}
In Section \ref{subse:lowcecontivr} we discussed the proof that 
there is a low c.e.\ degree containing an \ivr\ set. As in that argument, 
we make use of {\em levels} in the definition of the reductions $\Gamma$ and $\Delta$. We slightly adjust the definition of the levels because we now must also enumerate the markers $\lambda(e,s)$ into $B$. We increase the size of each interval $[d_i, d_{i+1})$ to accommodate a coding marker. The coding markers will be chosen to be $d_i$ for some $i \in \omega$. Now that we are also enumerating coding markers into $B$, we also adjust the definition of the levels $l_i$. Enumerating the coding markers will cause more strings to become {\em forbidden}. We recalculate the upper bound on the number of forbidden strings of length $l_{i+1}$. 
As before, the requirements $R_0, \ldots, R_i$ may act at level $l_{i+1}$. We calculated in 
Section \ref{subse:lowcecontivr} 
that we may act at most $(i+2)!$ many times at level $l_{i+1}$. When the coding marker $d_i$ is enumerated, the requirements $R_0, \ldots, R_i$ may act again at level $l_{i+1}$. Therefore we may act at most $2(i+2)! + 1$ many times at level $l_{i+1}$. So there are at most $ \sum_{j=0}^i (2(j+1)! + 1)$ many forbidden strings of length $l_{i+1}$. Letting $f'_{i+1} = \sum_{j=0}^i (2(j+1)! + 1)$, we calculate the levels $l_i$ as before. 
We let $l_0 = 0$, and given $l_i$, we let $l_{i+1,0} = l_i + (\mu n)(2^n \geqslant f'_{i+1})$ and for $0 \leqslant j \leqslant i$ 
we let
\[ 
l_{i+1,j+1} = l_{i+1,j} + \max_{k \leqslant i} \ n(2^{l_{i+1,j} + 1}m_k(\lambda)) 
\]
\noindent 
and  $l_{i+1} = l_{i+1,i+1} + 1$. 
Set $d_0 = 0$. Given $d_i$, let $d_{i+1} = d_i + 2(i+1)! + 1 + 1$. 

\subsubsection{Coordination and restraint on the tree}
The action for $\mathcal{R}$-requirements will otherwise be identical with 
the construction of Section \ref{subse:lowcecontivr}.
Nodes working for $\mathcal{R}$-requirements will also have to be wary of coding done by $\mathcal{P}$-nodes above. Suppose $\beta$, working for $\mathcal{R}_f$, is below the $\infty$ outcome of $\alpha$, which is working for $\mathcal{P}_e$. If we see the martingale $m_f$ increase its capital along $A_s$ and wish to enumerate an element of $[d_i+1, d_{i+1})$ into $B$ to change $A$, we will wait until all unrestrained markers $\lambda(e,s)$ below $d_i$ have been enumerated into $B$ before changing $A$ for the sake of $\mathcal{R}_f$.
If $\sigma$ is accessible at stage $s$, we let 
\[ 
r(\sigma,s) = \max \{ \sigma(|\alpha|)\ |\  \alpha \prec \sigma \mbox{ is an } 
\mathcal{N}\mbox{-node or an } \mathcal{R}\mbox{-node} \}.
\]
We say that a computation $\Phi_e^B(e)[s]$ is $\sigma$-correct 
if for every $\mathcal{P}$-node $\alpha$ such that 
$\alpha \ \hat{\empty} \ \infty \preccurlyeq \sigma$, 
$r(\alpha,s) < \lambda(e,t) < \varphi_e^B(e)[s]$ implies $\lambda(e,t) \in B_s$.
 Recall the definition of the {\em sibling} of a string, just after 
Definition \ref{de:intvalmar}.

\subsubsection{Construction of the sets \texorpdfstring{$A, B$}{A,B}}
At stage 0 we set $A_1 = 1^{\omega}$, $B_1 = \emptyset$, 
and let $r_{e,1} = l_e$ for all $e$. Moreover we 
set $\Lambda^B(0,0) = 0$ with use $d_1$.
Each stage $s \geqslant 1$ is conducted in two steps:

\textit{Step 1}: For $e,t \leqslant s$, if $\Lambda^B(e,t)$ is undefined at stage $s$, then let $\Lambda^B(e,t) = 1$ with some fresh large use $\lambda(e,t)$ equal to $d_i$ for some $i \in \omega$.

\textit{Step 2}: We let the collection of accessible nodes $\delta_s$ be an initial segment of the tree of strategies.
Let $\sigma$ be a node which is accessible at stage $s$. We describe the action that $\sigma$ takes, and if $|\sigma| < s$, then we specify which immediate successor of $\sigma$ is also accessible at stage $s$; otherwise, we proceed to the next stage. 

Suppose first that $\sigma$ works for $\mathcal{R}_e$. 
Let $k$ be least such that $d_k \geqslant r(\sigma,s)$. If 
\begin{enumerate}
\item
for all $\mathcal{P}$-nodes $\alpha$ such that 
$\alpha \ \hat{\empty} \ \infty \preccurlyeq \sigma$, 
$r(\alpha,s) < \lambda(e,t) < d_k \implies \lambda(e,t) \in B_s$, and 
\item
there is $l > l_{k+1}, r_{e,s}$ such that 

\begin{enumerate}
\item 
$m_{e,s}(\sigma)\downarrow$ for all strings $\sigma$ of length $l$, and 
\item
$m_e(A_s \upharpoonright l) > m_e(A_s \upharpoonright (l-1))$, 
\end{enumerate}

\end{enumerate}
then let $i$ be such that $l \in (l_i, l_{i+1}]$ (if there is more than one such $l$, we choose the least). 
Choose a string $\tau$ of length $l_{i+1}$ which extends the sibling of $A_s \upharpoonright (l-1)$ such that 
the minimum of all $d(\tau,\mu)$, where $\mu$ is any forbidden string of length 
$l_{i+1}$ extending $A_s \upharpoonright l_i$, is as large as possible.
Enumerate an element of $[d_i + 1,d_{i+1})$ into $B$. 
Choose a string $\rho$ of length $l_{i+2}$ extending $\tau$ such that 
$\rho \not\prec A_t$ for all $t < s$, and 
the minimum of all $d(\rho,\mu)$, where the minimum is taken over forbidden strings $\mu$ of length $l_{i+2}$ extending $\tau$, 
is as large as possible.
Set $A_{s+1} = \rho1^{\omega}$. For all $e' \geqslant e$ with $r_{e',s} \leqslant l_{i+1}$, let $r_{e',s+1} = l_{i+1}$.
The string $\sigma \ \hat{\empty} \ d_{i+1}$ is accessible at stage $s$.

Now suppose that $\sigma$ works for $\mathcal{P}_e$. For all $e,t \leqslant s$, if $\lambda(e,t) > r(\sigma,s)$, $e \not\in S_s$ and we have $\Lambda^B(e,t)[s] = 1$, then enumerate $\lambda(e,t)$ into $B$ and define $\Lambda^B(e,t) = 0$ with use $-1$. Suppose that $\lambda(e,t) = d_{i}$. 
Choose a string $\rho$ of length $l_{i+1}$ extending $A_s \upharpoonright l_i$ such that 
$\rho \not\prec A_t$ for all $t < s$, and 
the minimum of all $d(\rho,\mu)$, where the minimum is taken over forbidden strings $\mu$ of length $l_{i+1}$ extending 
$A_s \upharpoonright l_i$, is as large as possible.
Set $A_{s+1} = \rho \ \hat{\empty} \ 1^{\omega}$. If a marker was enumerated, let $\sigma \ \hat{\empty} \ \infty$ be accessible at stage $s$. Otherwise let $\sigma \ \hat{\empty} \ f$ be accessible at stage $s$.

Finally suppose that $\sigma$ works for $\mathcal{N}_e$. If $\Phi_e^B(e)[s]\downarrow$ via a $\sigma$-correct computation, then let $\sigma \ \hat{\empty} \ \varphi_e^B(e)[s]$ be accessible at stage $s$. Otherwise, let $\sigma \ \hat{\empty} \ 0$ be accessible at stage $s$.

\subsubsection{Verification of the construction}
By the construction, the set $B$ is c.e. We verify that
$A \equiv_T B$, that $A$ is ivr, and that $A'\equiv_T S$. First, we establish the existence of 
a `true path'.
\begin{equation}\label{eq:leftmverpathjinv}
\parbox{10cm}{The leftmost path which is visited infinitely often exists.}
\end{equation}

As there are only finitely many outcomes of a $\mathcal{P}$-node, we need to verify that the restraint imposed by an $\mathcal{N}$- or $\mathcal{R}$-node comes to a limit. Let $\sigma$ work for $\mathcal{R}_e$, and suppose by induction that no node to the left of $\sigma$ is visited after stage $s_0$, and that $\lim_s r(\sigma,s)$ exists. We must have $r(\sigma,s) = r(\sigma,s_0)$ and $r_{e,s} = r_{e,s_0}$ for all $s \geqslant s_0$. Let $k$ be such that $d_k \geqslant r(\sigma,s_0)$. Whenever we act for $\mathcal{R}_e$ at level $l_i$ for $i \geqslant k+1$, $m_e$'s capital decreases by at least \$1 and we increase restrains for all $\mathcal{R}$-requirements of weaker priority. The only elements which may be enumerated below the restraints $\mathcal{R}_e$ places on $B$ are coding markers belonging to $\mathcal{P}$-requirements stronger than $\mathcal{R}_e$. However, if $\mathcal{R}_e$ is below the $\infty$ outcome of $\mathcal{P}_e$, then $\mathcal{R}_e$ waits until all unrestrained markers below $d_i$ enter $B$ before acting at level $l_{i+1}$. Therefore the only markers which enter below $\mathcal{R}_e$'s restraint belong to those $\mathcal{P}$-requirements with $\sigma$ below the $f$ outcome. By induction, we do not visit any node to the left of $\sigma$ after stage $s_0$, and so no such strategy may act after stage $s_0$ and enumerate a coding marker below $\mathcal{R}_e$'s restraint. Therefore we act for $\mathcal{R}_e$ only finitely many times after stage $s_0$.
Similarly, as we require the computations $\mathcal{N}$-nodes observe to be $\sigma$-correct, if $\sigma$ works for $\mathcal{N}_e$ and is on the true path, it will increase its restraints only finitely many times.
This concludes the proof of \eqref{eq:leftmverpathjinv}.

Let the true path, TP, be the leftmost path visited infinitely often. 
The proof of \eqref{eq:leftmverpathjinv} 
shows that we act only finitely often for any $\mathcal{R}$-requirement. 
Therefore $m_e(A) < \infty$ and $\mathcal{R}_e$ is satisfied for all $e \in \omega$.
Therefore
\begin{equation}\label{eq:Aisivrathjinv}
\parbox{10cm}{the set $A$ is \ivr.}
\end{equation}
The use of the systems of levels $l_i, d_i$
in the construction define Turing reductions
$A\leq_T B$ and $B\leq_T A$ 
respectively, with computable use.
So that by an induction on the stages of the construction we have
\begin{equation}\label{eq:ABsameturngdegr}
\parbox{10cm}{the sets $A$ and $B$ are in the same weak truth table degree.}
\end{equation}
It remains to show that $S \equiv_T B'$.
First, we show that $S \leqslant_T B'$.
Let $\sigma \prec$ TP be devoted to $\mathcal{P}_e$, and suppose that no node to the left of $\sigma$ is visited after stage $s_0$. As the restraints set by $\mathcal{R}$- and $\mathcal{N}$-requirements are finite, $r(\sigma,s_0)$ is finite and $r(\sigma,s) = r(\sigma,s_0)$ for all $s \geqslant s_0$.  Therefore $\sigma$ may enumerate all but finitely many markers if it wishes. Therefore for all $e \in \omega$, 
$\lim_t \Lambda^B(e,t) = S(e)$.
It remains to show that
$B' \leqslant_T S$.
We have 
\begin{equation}
e \in B' \Leftrightarrow \Phi^B_e(e)\downarrow \Leftrightarrow (\exists t)(\Phi^B_e(e)[t]\downarrow \wedge B \upharpoonright \varphi_e^B(e)[t] = B_t \upharpoonright \varphi_e^B(e)[t]).
\end{equation}
First use $S$ to compute $S \upharpoonright e$. For $i < e$, if $S(i) = 1$ then we will want to eventually stop enumerating markers $\lambda(i,s)$. If $S(i) = 0$, then we will want to enumerate all unrestrained markers. Suppose we see a computation $\Phi^B_e(e)[t] \downarrow$. We find the markers $\lambda(i,s)$ below $\varphi^B_e(e)[t]$ for $i < e$. As we know the fate of every marker below the use, we can computably determine whether this computation is $B$-correct, that is, whether $B$ will change below the use after stage $t$. Therefore equation (1) is $\Sigma^0_1$, and can be decided by $\emptyset'$. As $\emptyset' \leqslant_T S$, we have $B' \leqslant_T S$. 

\subsection{Degrees of left-c.e.\ partial integer-valued randoms}\label{subse:delceintv2}
Here we show that every left-c.e.\ real that is integer-valued random is Turing (and in fact, weak truth table) 
equivalent to a partial integer-valued random  left-c.e.\ real. Hence along with the argument of the previous section, it
proves  Theorem \ref{cedegcontivrlceres}. 
In order to make the argument more concise, we will often refer to the argument of Section  \ref{subse:lowcecontivr}, which uses
a similar machinery.
Given a  left-c.e. \ivr\ set $A$ we will construct a partial \ivr\ set $B$ such that $A \equiv_{\textrm{wtt}} B$.
Suppose we are given $A$ with a computable approximation $\langle A_s \rangle$. Let $\langle \varphi_e \rangle$ be an effective enumeration of all partial computable integer-valued martingales. We build a set $B$ and weak truth-table reductions $\Gamma$ and $\Delta$ such that $\Gamma^A = B$ and $\Delta^B = A$ to meet the requirements 

\begin{description}
	\item[$\mathcal{R}_e$] $\varphi_e$ does not succeed on $B$.
\end{description}
\noindent
We also build for each $e \in \omega$ a partial integer-valued martingale $m_e$. In the case that $\varphi_e$ is defined along $B$, $m_e$ will be total.

\subsubsection{Strategy for the single requirement $\mathcal{R}_0$.} 
Let $\gamma(n)[s]$ be the use of computing $\Gamma^A(n)[s]$ and $\delta(n)[s]$ the use of computing $\Delta^B(n)[s]$. We begin by setting $\gamma(n)[0] = n$ and $\delta(n)[0] = n$ for all $n$. We observe the values of $\varphi_0$ along $B$. First we wait to see $\varphi_0(\lambda)$. If we later see $\varphi_0$ increase its capital along $B$, then we will wish to change $B$ to force $\varphi_0$ to lose capital. We will need permission from $A$ to do so. We put pressure on $A$ to change by defining the martingale $m_0$. If $\varphi_0$ increases its capital on $B_s$ by betting on $B_s \upharpoonright n$, then we define $m_0$ to start with capital $\varphi_0(\lambda)$, place the same bets as $\varphi_0[s]$ along $B_s \upharpoonright n$, and bet neutrally on all other strings up to length $n$. We repeat this every time we see $\varphi_0$ increase its capital along $B$ until we see a change in $A$. As $A$ is \ivr, $m_0(A) < \infty$ and so $A$ must eventually move. This gives us a permission to change $B$. 

Suppose that at stage $s$ we have defined $m_0$ up to length $d$, and $A$ changes below $d$. Let $m$ be least such that $A_{s-1}(m) \ne A_s(m)$. We have defined $\delta(m)[s-1] = m$ and so we must change $B$ on its $m$th bit. We let $B_s = B_{s-1} \upharpoonright m \ \hat{\empty} \ (1 - B_{s-1}(m)) \ \hat{\empty} \ 0^{\omega}$. The partial martingale $\varphi_0$ might not be defined on any string extending $B_s \upharpoonright m$, whereas $m_0$ has been defined to be neutral on all initial segments of $B_s$ of length between $m$ and $d$. Later $\varphi_0$ might increase its capital along these strings, and we would not be able to define $m_0$ to directly copy its bets. We can however raise the use $\gamma(m)[s]$ to be large. At stage $s$ we let $\gamma(m)[s] = 2d$. If $\varphi_0$ then bets along initial segments of $B$ of length between $m$ and $d$, we copy the wagers that $\varphi_0$ makes on these strings by placing the same wagers along the initial segments of $A$ of length between $d$ and $2d$. We are then still putting pressure on $A$ to change. If $A$ changes below $2d$ we can then change $B$ below $d$ and force $\varphi_0$ to lose money.

\subsubsection{Multiple requirements}
The interaction between multiple requirements will cause difficulty in coding. We use levels $\langle l_i \rangle$ and $\langle d_i \rangle$ in order to facilitate the coding. We set $\gamma(l_i) = d_i$ and $\delta(d_i) = l_{i+1}$ and let $l_1 = 5$ (the choice of $l_1$ is not significant). We attempt to change $B$ above a string of length $l$ for $l \in (l_i,l_{i+1}]$ only for the sake of the requirements $\mathcal{R}_0, \ldots, \mathcal{R}_i$. 

We will attempt to change $B$ below $l_1$ only for the sake of decreasing $\varphi_0$'s capital along $B \upharpoonright l_1$. We would like $\gamma(l_1)$ to be large enough so that we can copy all the wagers that $\varphi_0$ may place along strings of length $l_1$. There are $2^{l_1}$ many such strings, and so if we set $d_1 = \gamma(l_1) = 2^{l_1}.l_1$, this will certainly be large enough. As $A$ is left-c.e., $A$ can change below $d_1$ at most $2^{d_1}$ many times. Therefore there are at most $2^{d_1}$ many forbidden strings. To calculate $l_2$, we begin by setting $l_{2,0} = l_1 + (\mu x)(2^x \geqslant 2^{d_1}) = l_1 + d_1$. We act at level $l_2$ for the sake of $\mathcal{R}_0$ and $\mathcal{R}_1$. The action for these requirements can again move us closer to forbidden strings. The distance we can be moved, is given in terms of the function $d$ which is introduced in the
argument of Section \ref{subse:lowcecontivr}.
Therefore we define $l_{2,1}, l_{2,2}$ and $l_2$ as before.

Suppose $A$ changes below $d_1$ at stage $s$. Then we are free to change $B$ below $l_1$. We choose a string $\tau$ of length $l_1$ which minimises $m_0$; if there is no reason to move, we do not move. In either case, we then choose some string $\rho$ of length $l_2$ extending the current version of $B \upharpoonright l_1$ which has not been visited before, and let $B_{s+1} = \rho1^{\omega}$. 

We change $B$ below $l_2$ for the sake of requirements $\mathcal{R}_0$ and $\mathcal{R}_1$. We define the total martingale $m_0$ to copy the wagers that $\varphi_0$ places on $B \upharpoonright l_2$, and we define the total martingale $m_1$ to copy the wagers that $\varphi_1$ places on $B \upharpoonright l_2$. We require $d_2$ to be large enough so that $m_0$ can copy all the wagers that $\varphi_0$ may place along strings of length $l_2$ beneath $d_2$, and $m_1$ can copy all the wagers that $\varphi_1$ may place along strings of length $l_2$ beneath $d_2$. Therefore, by the same reasoning as the calculation of $d_1$, we would like $d_2$ to be at least $2^{l_2}.l_2$. We then define $l_3$ similarly, with $l_{3,0} = l_2 + d_2$, and $l_{3,1}, l_{3,2}, l_{3,3}$ and $l_3$ as 
in the argument of Section \ref{subse:lowcecontivr}.

Now suppose that $A$ changes between $d_1$ and $d_2$ at stage $s'$. That is, $A_{s'} \upharpoonright d_1 = A_{s'-1} \upharpoonright d_1$, but there is $m \in [d_1,d_2)$ with $A_{s'}(m) \ne A_{s'-1}(m)$. We cannot change $B$ below $l_1$, but we can change $B$ above $l_1$. We therefore choose a string $\tau$ of length $l_2$ which minimises the martingales $\varphi_0$ and $\varphi_1$. As $\mathcal{R}_0$ has stronger priority than $\mathcal{R}_1$, we first look to minimise $\varphi_0$. If we can, we change $B$ to minimise $\varphi_0$, and if we cannot, we then look to minimise $\varphi_1$. If we can, we change $B$ to minimise $\varphi_1$, and if we cannot, we do not change $B$. In any case, we then choose some string $\rho$ of length $l_3$ extending the current version of $B \upharpoonright l_2$ which has not been visited before, and let $B_{s'+1} = \rho1^{\omega}$. 

\subsubsection{Construction}
For every $e \in \omega$ and at every stage $s$ we have a restraint $r_{e,s}$. During the construction we will say that \textit{``$m_e$ has copied $\varphi_e$'s wager on $\sigma$''} for some $e$ and string $\sigma$. Let $d_{e,s}$ denote the length of longest string for which we have defined $m_e$ by the beginning of stage $s$. 
Set $l_0 = 0$, $d_0 = 0$ and $l_1 = 5$. Given $l_i$, we set $d_i = 2^{l_i}.l_i$, and then given $d_i$, we set $l_{i+1, 0} = l_i + d_i$ and for $0 \leqslant j \leqslant i$,
\[ l_{i+1,j+1} = l_{i+1,j} + \max_{k \leqslant i} \ n(2^{l_{i+1,j} + 1}m_k(\lambda)), \]
and let $l_{i+1} = l_{i+1,i+1} + 1$. Set $\gamma(l_i) = d_i$ and $\delta(d_i) = l_{i+1}$ for all $i$. At stage 0, we set $B_1 = 1^{\omega}$, $m_e(\lambda) = \varphi_e(\lambda)$ for all $e$, and $r_{e,1} = l_e$ for all $e$. At stage $s>0$ do the following: 

\textit{Case 1}: there is $l \geqslant r_{e,s}$ such that $\varphi_e(B_s \upharpoonright (l+1)) > \varphi_e(B_s \upharpoonright l)$, $m_e$ has copied $\varphi_e$'s wagers on $B_s \upharpoonright 1, \ldots, B_s \upharpoonright l$, and $m_e$ has not already copied $\varphi_e$'s wager on $B_s \upharpoonright (l+1)$. Let $e$ be the least applicable, and $l$ the least applicable for this $e$. Define $m_e$ to wager $\varphi_e(B_s \upharpoonright (l+1)) - \varphi_e(B_s \upharpoonright l)$ on $A_s \upharpoonright (d_{e,s} + 1)$ and wager 0 on all other strings of length $d_{e,s} + 1$. We say that $m_e$ has copied $\varphi_e$'s wager on $B_s \upharpoonright (l+1)$. Let $i$ be such that $l \in (l_i,l_{i+1}]$. For all $e' > e$ with $r_{e',s} \leqslant l_{i+1}$, let $r_{e',s+1} = l_{i+2}$. Proceed to the next stage.

\textit{Case 2}: $A_s(m) \ne A_{s+1}(m)$. Let $m$ be the least applicable, and let $i$ be such that $m \in (d_i,d_{i+1}]$. 
Choose the least $e$ with $r_{e,s} < l_{i+1}$ such that there is $\tau$ of length $l_{i+1}$ extending $B_s \upharpoonright l_i$ with $\max_{j \leqslant l_{i+1}} \varphi_{e,s+1}(\tau \upharpoonright j) < \max_{j \leqslant l_{i+1}}\varphi_{e,s+1}(B_s \upharpoonright j)$. For this $e$, choose an applicable string $\tau$ with $\min_\mu d(\tau,\mu)$ as large as possible, where the minimum is taken over all forbidden strings of length $l_{i+1}$ extending $B_s \upharpoonright l_i$. If there is no such $e$, then let $\tau = B_s \upharpoonright l_{i+1}$.
Choose a string $\rho$ of length $l_{i+2}$ extending $\tau$ such that $\rho \not\prec B_t$ for all $t \leqslant s$. Let $B_{s+1} = \rho1^{\omega}$. Proceed to the next stage.

If neither case applies, proceed to the next stage.

\subsubsection{Verification of the construction}
First, we observe that
$A \equiv_{wtt} B$. Indeed, the uses $\gamma(l_i)$ are clearly computable, 
and so we have $\Gamma^A = B$ via the weak truth-table functional $\Gamma$.
The consistency of $\Delta$ is a consequence of $B$ never extending a forbidden string. 
Again the uses $\delta(d_i)$ are computable and so $\Delta^B = A$ via the weak truth-table functional $\Delta$. 

It remains to show that
for all $e \in \omega$, $\mathcal{R}_e$ is satisfied. 
Suppose by induction that all requirements of stronger priority than $\mathcal{R}_e$ do not act after stage $s^*$. We show that if $\varphi_e$ succeeds on $B$, then $m_e$ succeeds on $A$, which is a contradiction to $A$ being \ivr. 

Let $i_0$ be least such that $l_{i_0} \geqslant r_{e,s^*}$. 
By the restraints imposed, we cannot change $B$ below $l_{i_0}$ for the sake of $\mathcal{R}_e$. 
Now suppose that at stage $s \geqslant s^*$ we see $\varphi_{e,s}(B_s \upharpoonright l) > \varphi_{e,s}(B_s \upharpoonright l_{i_0})$. Let $l$ be the least applicable, and suppose $i$ is such that $l \in (l_i,l_{i+1}]$. 
From stage $s$ we have $m_e$ copy $\varphi_e$'s wagers, and at some stage $s' \geqslant s$ we define $m_e$ to wager $\varphi_e(B_{s'} \upharpoonright l) - \varphi_e(B_{s'} \upharpoonright (l-1))$ on $A_{s'} \upharpoonright d_{e,s'} + 1$. 
If $m_e$ copies all of the wagers that $\varphi_e$ makes on strings of length less than or equal to $l-1$, then $m_e$ is defined on strings of length at most $2^{l-1}\cdot (l-1)$. Therefore $d_{e,s'} \leqslant 2^{l-1}\cdot (l-1) < d_{i+1}$. 

Suppose that at stage $t > s'$ we see $A$ change below $d_{i+1}$. Then we will choose a string $\tau$ of length $l_{i+1}$ extending $B_t \upharpoonright l_i$ with $\max_{j \leqslant l_{i+1}} \varphi_{e,t}(\tau \upharpoonright j) < \varphi_{e,t}(B_s \upharpoonright l)$. 
Taking the contrapositive, we see that if $\varphi_e$ makes capital on $B$ past $l_{i_0}$, 
then $m_e$ makes capital on $A$. Therefore if $\varphi_e$ succeeds on $B$ then $m_e$ succeeds on $A$.

\section{Computably enumerable degrees not containing \ivrs}
It is hardly surprising that there are c.e.\ degrees which do not contain \ivrs. After all, computable
enumerability hinders randomness, and indeed  with respect to a sufficient level
of randomness, c.e.\ sets are not random. 
However integer-valued randomness is sufficiently weak so that it has interesting interactions with
computable enumerability.
In this section we look at the question of which c.e.\ degrees
do not contain \ivrs.
The first example of such degrees was given in Section
\ref{subse:ancdegncomivr}
where we showed
Theorem  \ref{ancccc}. Perhaps more surprising is the fact that there are c.e.\ degrees
which are not array computable (so, by Corollary \ref{coro:anccpivr} they compute an \ivr)
yet they do not contain \ivrs. We prove this in Section \ref{subse:anrncontivr}, and extend it to
a much stronger result (namely Theorem \ref{th:high2cenoivr}) 
in Section \ref{subse:high2nocontivrs}.

\subsection{C.e.\ array noncomputable degrees not containing \ivrs}\label{subse:anrncontivr}
We wish to construct an array noncomputable c.e.\ degree not containing an \ivr\ set.
Let $(\Gamma_e, \Delta_e)_{e \in \omega}$ be an effective listing of all pairs of Turing functionals, and let $\langle D_n \rangle$ be the very strong array with $D_0 = \{ 0 \}, D_1 = \{ 1,2 \}, D_2 = \{ 3, 4, 5 \}, \ldots$. We build a c.e.\ set $B$ to satisfy the requirements

\begin{description}
	\item[$R_e$] $(\exists n)( W_e \cap D_n = B \cap D_n $)
	\item[$N_e$]  $\Delta_e^B = A_e   \wedge \Gamma_e^{A_e} = B \implies A_e$ is not \ivr.
\end{description}
We build for each $e \in \omega$ an integer-valued martingale $m_e$, and replace the requirement $N_e$ with the following requirements $N_{e,k}$ for all $k \geqslant 2$:

\begin{description}
	\item[$N_{e,k}$]  $\Delta_e^B = A_e   \wedge \Gamma_e^{A_e} = B \implies m_e$ wins at least $k$ dollars on $A_e$.
\end{description}
We effectively order the requirements, making sure that if $k < k'$, then $N_{e,k}$ has stronger priority than $N_{e,k'}$.
We say that $R_e$ requires attention at stage $s$ if 

\begin{enumerate}
\item
$R_e$ has no follower at stage $s$, or
\item
$R_e$ has follower $i$ at stage $s$ and $W_{e,s} \cap D_i \ne B_s \cap D_i$.
\end{enumerate}
For any $e \in \omega$, let $d_{e,s}$ be the length of the longest string $\sigma$ for which $m_e(\sigma)$ is defined by stage $s$. 
We have for every $e \in \omega$ a restraint $r_e$. 
Let $r_{e,k,s} = \max D_i$ where $i$ is the follower at stage $s$ of any $R$-requirement of stronger priority than $N_{e,k}$.
Let $l(e,s)$ be the length of agreement between $\Gamma_e^{\Delta_e^B}$ and $B$ at stage $s$, 
\[ 
l(e,s) = \max \{ x\ |\  (\forall y < x)(\Gamma_e^{\Delta_e^B}(y)[s] = B(y)[s] \}. 
\]
We say that $N_{e,k}$ requires attention at stage $s$ if 

\begin{enumerate}
\item
$m_e(A_{e,s} \upharpoonright d_{e,s}) = k-1$,
\item

$l(e,s) > r_{e,k,s}$,

\item

$l(e,s) > \delta_e(\gamma_e(r_{e,k,s}))[s]$

\item

$\gamma_e(l(e,s))[s] > d_{e,s}$

\end{enumerate}
Figure \ref{fig:high2neviz} illustrates the reductions involved in requirement $N_e$.

\begin{figure}
\begin{tikzpicture}
\draw[->,thick]  (0,0) -- (10,0);
\draw[->,thick] (0,2) -- (10,2);

\draw[<-,thick]  (2,0) -- (3,2); 	
\draw[<-,thick]  (3,2) -- (4,0); 
\draw[<-,thick]  (6,0) -- (7,2);
\draw[<-,thick]  (7,2) -- (8,0);

\filldraw [ball color=black] (2,0) circle (0.6pt);
\filldraw [ball color=black] (6,0) circle (0.6pt);
\node at (2,1) {\footnotesize $\Gamma_e$};
\node at (4,1) {\footnotesize $\Delta_e$};
\node at (6,1) {\footnotesize $\Gamma_e$};
\node at (8,1) {\footnotesize $\Delta_e$};
\node at (2,-0.5) {\footnotesize $r_{e,k,s}$};
\node at (3,2.5) {\footnotesize $\gamma_e(r_{e,k,s})[s]$};
\node at (4, -0.5) {\footnotesize $\delta_e(\gamma_e(r_{e,k,s}))[s]$};
\node at (6,-0.5) {\footnotesize $l(e,s)$};
\node at (7,2.5) {\footnotesize $\gamma_e(l(e,s))[s]$};
\node at (8,-0.5) {\footnotesize $\delta_e(\gamma_e(l(e,s)))[s]$};
\node at (10.5,0) {\footnotesize $B$};
\node at (10.5, 2) {\footnotesize $A_e$};
\end{tikzpicture} 
\caption{A visualization of the reductions in requirement $N_e$}
\label{fig:high2neviz}
\end{figure}
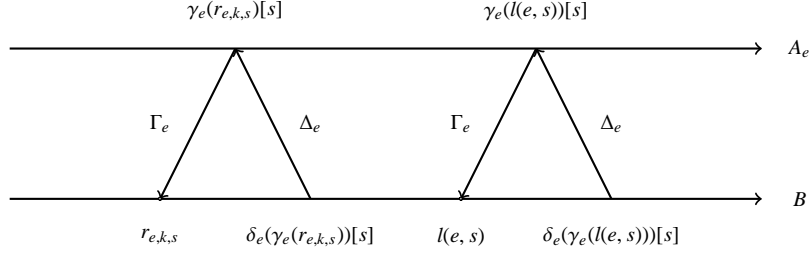

\ \paragraph{{\bf Construction}}
At stage 0, let $B_0 = \emptyset$. Let $m_{e}(\lambda) = 1$ for all $e \in \omega$.
At stage $s$, $s \geqslant 1$, find the requirement of strongest priority which requires attention at stage $s$. 

\textit{Case 1}: this is $R_e$. If $R_e$ has follower $i$, enumerate $W_{e,s} \cap D_i$ into $B$. If $R_e$ does not have a follower, appoint a fresh large follower for $R_e$.

\textit{Case 2}: this is $N_{e,k}$. Let $\tau = A_e[s] \upharpoonright \gamma_e(l(e,s))[s]$. Define $m_e$ to wager \$1 on $\tau$, and bet neutrally on all other strings with length in $(d_{e,s}, |\tau|]$.
Remove the followers of $R$-requirements of weaker priority than $N_{e,k}$.

\ \paragraph{{\bf Verification}}
It remains to show that
each requirement requires attention only finitely often, and is met.
Assume by induction that all requirements of stronger priority than $R_e$ do not require attention after stage $s$. If $R_e$ does not have a follower at stage $s$ then it will be appointed one. This follower cannot be cancelled as requirements of stronger priority can no longer act. Suppose that $R_e$ has follower $i$ at stage $s$. If we ever see that $B_t \cap D_i \ne W_{e,t} \cap D_i$ then we will enumerate $W_{e,t} \cap D_i$ into $B$. As $|D_i| = i+1$, $R_e$ can require attention at most $i+1$ many times after stage $s$. Then we will have $B \cap D_i = W_{e} \cap D_i$ and $R_e$ is satisfied.

We claim that for all $e$, if $\Delta_e^B = A_e$ and  $\Gamma_e^{A_e} = B$, then $m_e$ is nondecreasing along $A_e$. Suppose we have $m_e(\sigma) = k-1$ for some string $\sigma$. Suppose $N_{e,k}$ requires attention at stage $t$ and  we define $m_e(\tau) = k$. We remove the followers for $R$-requirements of weaker priority, and so only requirements of stronger priority than $N_{e,k}$ can enumerate elements into $B$ that are below $\delta_e(\gamma_e(l(e,t))[t]$. Suppose $R_j$ has stronger priority than $N_{e,k}$. Then $R_j$ can either enumerate elements below $r_{e,k,t}$, or if it is later injured, enumerate elements larger than $\delta_e(\gamma_e(l(e,t))[t]$ into $B$. Suppose that $R_j$ enumerates an element below $r_{e,k,t}$ into $B$.
If $\Gamma_e$ and $\Delta_e$ later recover at stage $t'$, $\tau' := A_e[t'] \upharpoonright \gamma_e(r_{e,k,t})[t]$ must be incomparable with $\tau = A_e[t] \upharpoonright \gamma_e(r_{e,k,t})[t]$; otherwise $\Delta_e$ will not have recorded the $B$-change and we could not have $\Delta_e^B = A_e$ and  $\Gamma_e^{A_e} = B$. In particular, $\tau'$ must not extend either $\tau$ or its sibling
(recall the definition of the {\em sibling} of a string, just after 
Definition \ref{de:intvalmar}). As $m_e$ is defined to bet neutrally on strings of length $|\tau|$ that are not either $\tau$ or its sibling, we will have $m_e(A_e[t'] \upharpoonright |\tau|) \geqslant k-1$. By induction, this holds for all $k$. This establishes the claim.

Now assume by induction that all requirements of stronger priority than $N_{e,k}$ do not require attention after some stage $s$. As $N_{e,j}$ for any $j < k$ does not require attention, we must have $m_e(A_{e,s} \upharpoonright d_{e,s}) = k-1$. If $N_{e,k}$ does not require attention at any stage $t > s$ then the hypothesis of the requirement does not hold. Therefore $N_{e,k}$ is satisfied vacuously. If $N_{e,k}$ does require attention at stage $t > s$ then we define $m_e(\tau) = k$ for $\tau = A_e[t] \upharpoonright \gamma_e(l(e,t))[t]$. $R$-requirements of stronger priority have finished acting, and so no numbers less than $r_{e,k,t}$ enter $B$ after stage $t$. We remove the followers for $R$-requirements of weaker priority. When they are appointed new followers they will choose fresh numbers, and so all enumerations into $B$ after stage $t$ will be larger than $\delta_e(\gamma_e(l(e,t))[t]$. As $B$ cannot change below $\delta_e(\gamma_e(l(e,t)))[t]$, $A_e = \Delta_e^B$ cannot change below $\gamma_e(l(e,t))[t]$. Therefore $\tau \prec A_e$ and $m_e(A_e \upharpoonright \gamma_e(l(e,t))[t] ) = k$. By the previous claim, $m_e(A_e) \geqslant k$, and $N_{e,k}$ is satisfied.

\subsection{A \texorpdfstring{high$_2$}{high-2} c.e.\ degree not containing \ivrs}\label{subse:high2nocontivrs}
In this section we prove Theorem \ref{th:high2cenoivr}.
Nies, Stephan and Terwijn showed that every high degree contains a computably random set, and so a fortiori, an \ivr\ set. It is instructive though to see why we cannot build a high degree that does not contain an \ivr\ set. This will give us some insight as to why the construction works when we only require our set to be high$_2$.
So that the degree of $A$ does not contain an \ivr\ set, we meet the requirements 

\begin{description}
	\item[$\mathcal{N}_e$]  $\Gamma_e^{\Delta_e^{A}} = A$ total $\implies \Delta_e^A$ is not \ivr
\end{description}

\noindent where $(\Gamma_e,\Delta_e)$ is an enumeration of pairs of Turing functionals. We break the requirement $\mathcal{N}_e$ into the following subrequirements $\mathcal{N}_{e,k}$. 

\begin{description}
	\item[$\mathcal{N}_{e,k}$]  $\Gamma_e^{\Delta_e^A} = A$ total $\implies m_{e}$ wins at least $k$ dollars on $\Delta_e^A$
\end{description}

\noindent where $m_e$ is an integer-valued martingale we build for the sake of $\mathcal{N}_e$. Suppose that the martingale $m_e$ has won \$$k-1$ on $\Delta_e^A \upharpoonright n$. The basic strategy to win another dollar is to first pick a large location marker $p$. Let $l(e)[s]$ be the length of agreement between $A$ and $\Gamma_e^{\Delta_e^A}$ at stage $s$. If we later see that $l(e)[s] > p$ and $l(e)[s] > \delta_e(\gamma_e(p))[s]$, then we define $m_e$ to wager \$1 on $\Delta_e^A[s] \upharpoonright \gamma_e(l(e))[s]$ and bet neutrally elsewhere, and freeze $A$ below $\delta_e(\gamma_e(l(e))[s]$. If we are successful in freezing $A$, then $m_e$ wins \$1 on $\Delta_e^A$. If $A$ changes below $p$ at stage $s' >s$, then if we are to have $\Gamma_e^{\Delta_e^A} = A$, $\Delta_e^A[s'] \upharpoonright \gamma_e(p)[s]$ is incomparable with $\Delta_e^A[s] \upharpoonright \gamma_e(p)[s]$ and $m_e$ does not lose any capital along $\Delta_e^A$. We can then try the basic strategy again.

\begin{figure}
\begin{tikzpicture}
\draw[->,thick]  (0,0) -- (10,0);
\draw[->,thick]  (0,2) -- (10,2);
\draw[<-,thick]  (2,0) -- (3,2);
\draw[<-,thick]   (3,2) -- (4,0);
\draw[<-,thick]  (6,0) -- (7,2);
\draw[<-,thick]  (7,2) -- (8,0);

\node at (2,1) {\footnotesize $\Gamma_e$};
\node at (4,1) {\footnotesize $\Delta_e$};
\node at (6,1) {\footnotesize $\Gamma_e$};
\node at (8,1) {\footnotesize $\Delta_e$};
\node at (2,-0.5) {\footnotesize $p$};
\node at (3,2.5) {\footnotesize $\gamma_e(p)[s]$};
\node at (4, -0.5) {\footnotesize $\delta_e(\gamma_e(p))	[s]$};
\node at (6,-0.5) {\footnotesize $l(e,s)$};
\node at (7,2.5) {\footnotesize $\gamma_e(l(e))[s]$};
\node at (8,-0.5) {\footnotesize $\delta_e(\gamma_e(l(e)))[s]$};
\node at (10.5,0) {\footnotesize $A$};
\node at (10.5, 2) {\footnotesize $\Delta_e^A$};
\end{tikzpicture}
\caption{Diagram with the uses and where we bet, for the proof of Theorem \ref{th:high2cenoivr}.}
\label{fig:diaghighforma}
\end{figure}
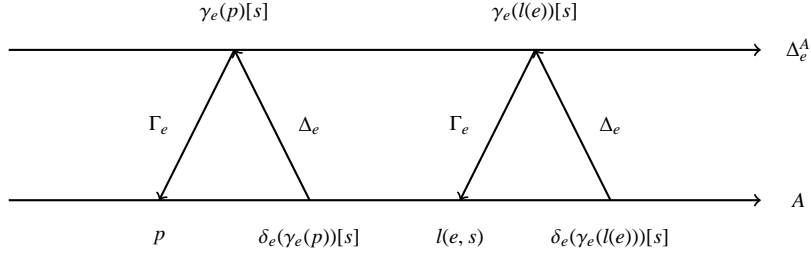

To make $A$ high we would define a functional $\Lambda$ such that $\lim_k \Lambda^A(x,k) =$ Tot$(x)$, where Tot is the canonical $\Pi^0_2$-complete set. The basic strategy for the highness requirement is to define $\Lambda^A(x,s) = 0$ for larger and larger $s$ with some big use $\lambda(x,s)$. When we see $\varphi_x(s') \downarrow$ for all $s' \leqslant s$, then for each $s' \leqslant s$ we enumerate the current use $\lambda(x,s')$ into $A$ (if currently $\Lambda^A(x,s') = 0$) and redefine $\Lambda^A(x,s') = 1$ with use $-1$, i.e. the axiom defining $\Lambda^A(x,s') = 0$ does not depend on $A$. This strategy will succeed as long as we are prevented from redefining $\Lambda^A(x,s)$ from 0 to 1 at most finitely often.

Let us see how these strategies might interact. Suppose at stage $s$ we saw the $l(e)$ computations converge and defined more of $m_e$. The highness requirement, if unrestrained, can destroy the $l(e)$ computations and cause $m_e$ to lose capital if it enumerates a marker between $p$ and $\delta_e(\gamma_e(l(e))[s]$. Therefore if $m_e$ is to ever win money along $\Delta_e^A$ we must impose restraint on $A$.
The problem is that the strategies for $\mathcal{N}_e$ may gang up and impose restraint on \emph{all} markers $\lambda(x,s)$; every time a marker is defined we may define $m_e$ and impose restraint, and never allow the marker to be enumerated. If Tot$(x) = 1$ we will never be able to correct $\Lambda^A(x,s)$ to be 1 and the limit will be incorrect.

Another approach we might take would be to capriciously enumerate the markers which occur below $\delta_e(\gamma_e(l(e))[s]$. If we do this and always have some marker below the use, we will be able to conclude that $\Gamma_e^{\Delta_e^A} \ne A$ (this argument is given in full in the verification below). However, the problem now is that we might not have wanted to enumerate the markers. If Tot$(x) = 0$ and we capriciously enumerate all markers $\lambda(x,s)$ and redefine $\Lambda^A(x,s)$ to be 1, the limit will be incorrect.

To make $A$ high$_2$, we need to instead define a functional $\Lambda$ such that 
\[
\lim_m \lim_t \Lambda^A(x,m,t) = \textrm{Cof}(x)
\] 
where Cof $ = \{ x\ |\  W_x \mbox{ is cofinite} \}$ is the canonical $\Sigma^0_3$-complete set. The double limit means that we may be wrong on a finite number of the $m$ while still satisfying the requirement. This will allow us to employ the capricious enumeration strategy successfully.
We have the requirements 

\begin{description}
	\item[$\mathcal{H}_x$]  $\lim_m \lim_t \Lambda^A(x,m,t) = $ Cof$(x)$
\end{description}

\noindent as well as the $\mathcal{N}_e$ from above. The construction will use a priority tree. For each global requirement $\mathcal{N}_e$ we have several nodes devoted to meeting $\mathcal{N}_e$, each equipped with a guess as to the outcomes of stronger priority requirements. Such a node will be called a {\em mother node} devoted to $\mathcal{N}_e$. Each such node $\tau$ builds its own martingale $m_\tau$. We argue in the verification that for every $e \in \omega$ there is some node $\tau$ such that $m_\tau$ succeeds on $\Delta_e^A$. For each subrequirement $\mathcal{N}_{e,k}$ we have several nodes devoted to meeting $\mathcal{N}_{e,k}$. Such a node will be called a {\em worker node} devoted to $\mathcal{N}_{e,k}$, and will occur below a mother node $\tau$ devoted to $\mathcal{N}_e$.
For the longest such $\tau$, we say that $\sigma$'s mother node is $\tau$. When we reach a node $\sigma$ devoted to $\mathcal{N}_{e,k}$, we 
choose a fresh location marker $p$ for $\sigma$, and
place a {\em link} from $\sigma$ back up to $\tau$. 
The link can be seen as testing the hypothesis of $\mathcal{N}_e$.
The length of agreement between $A$ and $\Gamma_e^{\Delta_e^A}$ will be measured at $\tau$. When we next arrive at $\tau$ and see the length of agreement computations converge, this further confirms the hypothesis of $\mathcal{N}_e$.
We travel the link to $\sigma$, define more of the martingale $m_\tau$, and then remove the link.

The requirement $\mathcal{H}_x$ will have nodes $\beta_{x,m}$ for $m \in \omega$. The node $\beta_{x,m}$ tests whether $[m, \infty) \subseteq W_x$. Note that this is a $\Pi_2$ test. Such a node will have outcomes $\infty$, which corresponds to the $\Pi_2$ test infinitely often looking correct, and $f$, corresponding to the finite outcome. The $\beta_{x,m}$ nodes will be responsible for defining $\Lambda^A(x,m,t)$ for each $t$. 
As we do not know the true path in advance, each path through the priority tree will have a $\beta_{x,m}$ node, which collectively will define $\Lambda^A(x,m,t)$ for all $t$. 

The basic strategy for $\beta_{x,m}$ is to define $\Lambda^A(x,m,t) = 0$ for larger and larger $t$ with some big use $\lambda(x,m,t)$. If we see $[m,s] \subseteq W_x$ then for each $s' \leqslant s$ we enumerate the current use $\lambda(x,m,s')$ into $A$ (if currently $\Lambda^A(x,m,s') = 0$) and redefine $\Lambda^A(x,m,s') = 1$ with use $-1$, i.e. the axiom defining $\Lambda^A(x,m,s') = 0$ does not depend on $A$. It is important to note that the markers $\lambda(x,m,s)$ are {\em shared} by all the $\beta_{x,m}$ nodes. Whether we succeed in enumerating the markers and updating $\Lambda$ as needed will depend on how the construction proceeds.



We will now describe how these requirements interact and what modifications we need to make to the priority tree as a consequence. First we consider the situation where we have $\mathcal{H}_x$ of lower priority than $\mathcal{N}_e$ (which is associated with the mother node $\tau$). The problem is the following. Suppose we have a situation with nodes $\tau \prec \beta_{x,m} \prec \sigma$ where $\sigma$ is devoted to $\mathcal{N}_{e,k}$ and $\tau$ is $\sigma$'s mother node. That is, while $\sigma$ has higher \emph{global} priority than $\beta_{x,m}$, its \emph{local} priority is lower. Suppose at some stage $\sigma$ picks a location marker $p(\sigma)$ and creates a link back to $\tau$ at stage $s_0$. At a later stage $s_1$ we get to $\tau$, see the necessary computations converge, and would like to travel the link and define the martingale. This causes no problem if $\sigma \succcurlyeq \beta_{x,m} \ \hat{\empty} \ f$, but there are problems if $\sigma \succcurlyeq \beta_{x,m} \ \hat{\empty} \ \infty$. We will require the computations $l(\tau)[s]$ to be $\tau$-correct; that is, all guesses $\tau$ makes about the enumeration of markers below $\delta_e(\gamma_e(l(e)))[s]$ have already occurred.


The trouble is that at stage $s_1$ there now might be some marker $\lambda = \lambda(x,m,q)$ which is greater than $p(\sigma)$, but below the use of $l(\tau)[s_1]$. We may not yet want to enumerate $\lambda$ into $A$ because the $\Sigma_2$ outcome may now be looking correct at $\beta_{x,m}$ (that is, we might think Cof$(x) = 0$). If we did define the martingale, since $\beta_{x,m}$ has higher priority than $\sigma$, any restraint imposed at $s_1$ may not be successful since $\beta_{x,m}$ might later put $\lambda$ into $A$. This could potentially cause our martingale $m_\tau$ to lose all its capital, and it could never bet again.


The solution to this problem is as follows. When we hit $\tau$, if there is some link to a node $\sigma$ and there is some $\lambda$ as above, we immediately enumerate any $\lambda$ below the use of $l(\tau)[s]$ into $A$, but we do \emph{not} define the martingale. This means that $\beta_{x,m}$ cannot later use $\lambda$ to make $m_\tau$ lose capital. If there is no such $\lambda$ then we do define the martingale, since we can be sure that $\sigma$ is satisfied provided that it is on the true path. This is the situation we would like, but failing that, we would like to get a global win on $\mathcal{N}_e$. In the case that such a $\lambda$ exists, we travel the link from $\tau$ to $\sigma$, enumerate all applicable markers, but we do \emph{not} delete the link. Because of this, we need to add a new outcome to $\sigma$. Therefore $\sigma$ will have outcomes $g$ and $d$. The outcome $g$ will be played when we perform the capricious enumeration of $\lambda$ as above. The outcome $d$ will be played when we define the martingale. 
Suppose we have some worker node $\sigma$ with location marker $p(\sigma)$ which always has some $\lambda$ below the use of $l(\tau)[s]$. We will then define $\Lambda^A(x,m,t)$ to have limit 1 for any pair $(x,m)$ such that $\tau \prec \beta_{x,m} \prec \sigma$; note that there are only finitely many pairs $(x,m)$. We will also have a link from $\tau$ to $\sigma$ for almost all stages. This corresponds, however, to a global win on $\mathcal{N}_e$, since $p(\sigma)$ is a witness to the fact that 
$\delta_e(\gamma_e(l(e))$ does not exist and so $\Gamma_e^{\Delta_e^A} \ne A$. 
The permanent link may cause us to skip over other mother nodes, which would mean we cannot meet their requirements. We therefore {\em restart} all $\mathcal{N}$-requirements of weaker priority than $\mathcal{N}_e$ under the $g$ outcome of $\sigma$. We do this by assigning the requirements $\mathcal{N}_{e'}$ for $e' >e$, as well as their subrequirements $\mathcal{N}_{e',k}$ for $k \in \omega$, to nodes below $\sigma \ \hat{\empty} \ g$ in some fair way.
We do not restart any $\beta$ nodes, since $\Lambda^A(x,m,t)$ will be defined to be 1 for the finitely many $x$ and $m$ with $\tau \prec \beta_{x,m} \prec \sigma$. This will mean that for a finite number of $m$, $\lim_t \Lambda^A(x,m,t)$ may be incorrectly outputting 1 instead of 0. This is fine though, since we will only lose on a finite number of the $m$ and still can satisfy $\mathcal{H}_x$. The sacrifice of losing on an $m$ will only be made when we can ensure a global win on a $\tau$ node of stronger priority.

We now come to the situation where we have an $\mathcal{H}_x$ of \emph{higher} global priority than the $\mathcal{N}_e$ associated with $\tau$. We now cannot allow $\tau$ to capriciously enumerate all the markers belonging to $\beta_{x,m}$ if $f$ is $\beta_{x,m}$'s true outcome. We now describe our solution to this problem.

First suppose that $\beta_{x,m}$ and $\beta_{x,n}$ are worker nodes devoted to $\mathcal{H}_x$ with $m < n$. If $\beta_{x,n}$ occurs below $\beta_{x,m} \ \hat{\empty} \ \infty$, then $\beta_{x,n}$ is guessing that $[m,\infty) \subseteq W_x$. Therefore $\beta_{x,n}$ must also be guessing that $[n,\infty) \subseteq W_x$ as $n > m$, and so $\beta_{x,n}$ will have only the $\infty$ outcome.

Suppose that $\tau$ is below the $f$ outcome of any $\beta_{x,m}$ node with $m< n$. We will restart $\tau$ below $\beta_{x,n} \ \hat{\empty} \ \infty$. Consider the situation $\beta_{x,n} \ \hat{\empty} \ \infty \preccurlyeq \tau \prec \beta_{x,n'} \prec \sigma$ with $n' > n$. If there is a link from $\tau$ to $\sigma$ then capriciously enumerating the markers $\lambda(x,n',t)$ into $A$ will not injure $\beta_{x,n'}$ since this is what $\beta_{x,n'}$ would like to do anyway. Therefore $\mathcal{H}_x$ cannot be injured in this situation. We show in Lemma \ref{finiteinjury} that such a $\tau$ can be restarted only finitely many times.

There is one last problem. Suppose we have $\tau \prec \beta_{x,m} \prec \sigma$ with a link $(\tau,\sigma)$. As $\sigma$'s mother is above $\beta_{x,m}$, we must have $\beta_{x,m} \ \hat{\empty} \ f \preccurlyeq \sigma$. If the link is permanent then $\beta_{x,m}$ will not be able to enumerate its markers. As $\mathcal{H}_x$ has higher global priority than $\mathcal{N}_e$, this is not a situation we want. We employ the following technique from Downey-Stob \cite{Downey.Stob:scout}. When we hit $\tau$, we realize that if there is a link from $\tau$ down then this may be a potentially permanent link. We first perform a {\em scouting report} to see where we would go if there were no link around. If we were to go to a node $\gamma$ to the left of $\sigma$ then we will erase the link and actually go to $\gamma$ instead. This ensures that if $\beta_{x,m} \ \hat{\empty} \ \infty$ is $\beta_{x,m}$'s true outcome, then it will be able to enumerate its markers.

\subsection*{The Priority Tree}
Our priority tree, PT, will have three types of nodes. The first type are \emph{mother} nodes $\tau$, which have outcomes $\infty$ and $f$, and will be assigned to some global requirement $\mathcal{N}_e$. We write $e(\tau) = e$. The next type are \emph{worker} nodes $\sigma$ which are devoted to a subrequirement of some $\mathcal{N}_e$, and hence will be assigned some $e,k$. We write $e(\sigma) = e, k(\sigma) = k$. We form the tree so that such $\sigma$ occur below some $\tau$ with $e(\tau) = e$. For the longest such $\tau$ with $e(\tau) = e$, we will write $\tau(\sigma) = \tau$. This is to indicate that $\tau$ is $\sigma$'s mother. $\sigma$ has outcomes $g$ and $d$ with $g < d$. Finally we have nodes $\beta$ which are devoted to some $\mathcal{H}_x$, and hence will be assigned some $x,m$. We write $x(\beta) = x, m(\beta) = m$, or simply $\beta_{x,m}$. $\beta$ will have outcomes $\infty$ and $f$, with $\infty < f$, unless $\beta$ occurs below $\beta' \ \hat{\empty} \ \infty$ for some $\beta'$ with $x(\beta') = x(\beta)$, in which case $\beta$ has only the single outcome $\infty$.

We now assign requirements and subrequirements to nodes on the tree. In a basic infinite injury argument we would have all nodes of the same level working for the same requirement. However in our case, as we must restart $\tau$ nodes, it is more complicated. We use {\em lists} of, for example, Soare \cite[Chapter XIV]{MR882921} for this. We will have three lists, $L_0, L_1$, and $L_2$, which keep track of indices for $\tau, \sigma$ and $\beta$ nodes, respectively.

\smallskip

$n = 0$. Let $\lambda$ be devoted to $\mathcal{N}_0$, and let $L_0(\lambda) = L_1(\lambda) = L_2(\lambda) = \omega$.

For $n > 0$, let $\gamma \in PT$ be of the form $\delta \ \hat{\empty} \ a$. Adopt the first case below to pertain, letting $L_i(\gamma) = L_i(\delta)$ unless otherwise mentioned.

\emph{Case 1}. $\delta$ is devoted to $\mathcal{N}_e$.

\emph{Case 1a}. $a = f$. $L_0(\gamma) = (L_0(\delta) - \{ e \}) \cup \{ e'\ |\  e' > e \}$
\[
L_1(\gamma) = (L_1(\delta) - \{ \langle e,k \rangle\ |\  k \in \omega \}) \cup \{ \langle e', k \rangle\ |\  e' > e, k \in \omega \}.
\]

\emph{Case 1b}. $a = \infty$. $L_0(\gamma) = L_0(\delta) - \{ e \}$. 

\emph{Case 2}. $\delta$ is devoted to $\mathcal{N}_{e,k}$. 

\emph{Case 2a}. $a = g$. Define the lists as in Case 1a.


\emph{Case 2b}. $a = d$. Let $L_1(\gamma) = L_1(\delta) - \{ \langle e,k \rangle \}$.

\emph{Case 3}. $\delta$ is devoted to $\mathcal{H}_x$ with $m(\delta) = m$. 

\emph{Case 3a}. $a = f$. $L_2(\gamma) = L_2(\delta) - \{ \langle x,m \rangle \}.$

\emph{Case 3b}. $a = \infty$. 
Let $L_0(\gamma)$ be the union of $ L_0(\delta)$ with 
\[ 
\{ e\ |\  (\exists \tau)(\forall \beta)(e(\tau) = e \wedge x(\beta) = x \wedge x < e \wedge \beta \prec \tau \prec \delta \implies \beta \ \hat{\empty} \ f \preccurlyeq \tau) \} 
\]
and let $L_1(\gamma)$ be the union of $L_0(\delta)$ with
\[ 
\{ \langle e,k \rangle\ |\  (\exists \tau)(\forall \beta)(e(\tau) = e \wedge x(\beta) = x \wedge x < e \wedge \beta \prec \tau \prec \delta \implies \beta \ \hat{\empty} \ f \preccurlyeq \tau), k \in \omega \}. 
\]
Also let $L_2(\gamma) = L_2(\delta) - \{ \langle x, m \rangle \}$.

Having defined the lists, we now assign requirements to nodes of the priority tree as follows. Let $\gamma \in$ PT and $i$ be the least element of $L_0(\gamma) \cup L_1(\gamma) \cup L_2(\gamma)$. If $i \in L_0(\gamma)$, let $\gamma$ be a mother node devoted to $\mathcal{N}_i$. If $i \in L_1(\gamma) - L_0(\gamma)$ and $i = \langle e,k \rangle$, let $\gamma$ be a worker node devoted to $\mathcal{N}_{e,k}$. Otherwise $i = \langle x,m \rangle \in L_2(\gamma)$ and we let $\gamma$ be worker node devoted to $\mathcal{H}_x$ with $m(\gamma) = m$.

\begin{lem}[Finite injury along any path lemma]\label{finiteinjury}
For every path $h \in [PT]$ and every $e,k \in \omega$,

\begin{enumerate}
\item
$(\exists^{< \infty} \alpha \prec h)(e(\alpha) = e \wedge h(|\alpha|) = g)$,
\item
$(\exists^{< \infty} \alpha \prec h)(\alpha  \mbox{ devoted to }  \mathcal{N}_e)$,
\item
$(\exists^{< \infty} \alpha \prec h)(\alpha  \mbox{ devoted to }  \mathcal{N}_{e,k})$.
\end{enumerate}

\end{lem}
\begin{proof}
(1) and (3) are standard. For (2), a node $\tau$ devoted to $\mathcal{N}_e$ is restarted below $\beta \ \hat{\empty} \ \infty$ if $x(\beta) < e$ and $\tau$ has been below only the $f$ outcomes of nodes devoted to $\mathcal{H}_x$. Once restarted, it can no longer be restarted below any other $\beta'$ with $x(\beta') = x(\beta)$. Thus $\tau$ is restarted finitely many times. 
\end{proof}

The construction below will proceed in substages. We will append a subscript $t$ to a parameter $G$, so that $G_t$ denotes the value of $G$ at substage $t$ of the construction. As usual all parameters hold their value unless they are initialized. When initialization occurs they become undefined, or are set to zero as the case may be. We will append a parameter $[s]$, when necessary, to denote stage $s$. We may write $(s, t)$ to denote substage $t$ of stage $s$.

If we visit a node $\nu$ at stage $(s, t)$ we will say that $(s, t)$ is a genuine $\nu$-stage. It might be that we do not visit $\nu$ at stage $(s, t)$, rather we visit some $\nu'$ extending $\nu$. In this case we say that $(s,t)$ is a $\nu$-stage, and hence a $\nu$-stage may not be genuine. In fact, should we put in place some permanent link $(\tau,\sigma)$ with $\tau \prec \nu \prec \sigma$, then $\nu$ might only ever be visited finitely often. However, this is when $\sigma \ \hat{\empty} \ g$ is the true outcome for some higher priority $\tau$, and we would claim that a new version of $\nu$ would live below outcome $g$ of $\sigma$. We will eventually define the genuine true path as those nodes that are on the leftmost path visited infinitely often, and for which there are infinitely many genuine stages.

We have for each node on the priority tree $\gamma$ and for all $x,m \in \omega$ a restraint $c_{\gamma}(x,m)$. These restraints will be initially set to zero in the construction, and will only be increased when $\gamma$ is a mother node devoted to some $\mathcal{N}_e$.
We say that a computation $\Xi^A(x)[s]$ is $\tau$-correct if 
$(\forall q \leqslant s)(\forall \beta)$ the condition
\[
(\beta \ \hat{\empty} \ \ \infty \preccurlyeq \tau \wedge x(\beta) = x \wedge m(\beta) = m \ \wedge 
\ \max_{\tau' \leqslant \beta} c_{\tau'}(x,m)[s] < \lambda(x,m,q) < \xi^A(x)[s])
\]
implies  $\lambda(x,m,q) \in A[s]$.
If $\tau$ is a mother node devoted to $\mathcal{N}_e$, then let 
\[ 
l_1(\tau)[s] = \max \{ x\ |\  (\forall y < x)(\Gamma_e^{\Delta_e^A}(y)[s] = A[s](y)) \ \mbox{via a $\tau$-correct computation}\}.
\] 
Let $m_1(\tau)[s] = \max \{ l_1(\tau)[s']\ |\  \mbox{ $s'<s$ is a genuine $\tau$-stage} \}$.
If $\tau$ is a node devoted to $\mathcal{H}_x$ with $m(\tau) = m$, then let $l_2(\tau)[s] = \max \{ y\ |\  [m,y] \subseteq W_x[s] \}$. Let 
\[
m_2(\tau)[s] = \max \{ l_2(\tau)[s']\ |\  \mbox{ $s'<s$ is a genuine $\tau$-stage} \}. 
\]
For $\tau$ a mother node devoted to some $\mathcal{N}_e$, let $d_{\tau}[s]$ denote the length of the longest string for which $m_{\tau}$ is defined by stage $s$.

\ \paragraph{{\bf Construction}}
At stage $0$ set $\Lambda^A(0,0,0) = 0$ with some large use $\lambda(0,0,0)$. Set $c_{\gamma}(x,m)[0] = 0$ for all $x, m \in \omega$ and all nodes $\gamma$ on the priority tree. Set $m_{\tau}(\lambda) = 1$ for all nodes $\tau$ devoted to some requirement $\mathcal{N}_e$.
Stage $s + 1$ will proceed in substages $t \leqslant s$. As usual, we will generate a set of accessible nodes, TP$[s + 1]_t$, and will automatically initialize nodes $\alpha$ to right of TP$[s + 1]_t$. A node is initialized by removing its location marker and removing any link to or from the node.

\emph{Substage 0}. Define TP$[s + 1]_0 = \lambda$, the empty string. Let $\Lambda^A(x,m,s+1) = 0$ with some large use $\lambda(x,m,s+1)$ for all $x,m \leqslant s+1$.

\emph{Substage $t + 1 \leqslant s+1$}. We will be given a string $\gamma =$ TP$[s + 1]_t$. Adopt the first case to pertain below.

\bigskip

\emph{Case 1}. $\gamma$ is a mother node devoted to $\mathcal{N}_e$.

\emph{Subcase 1a}. There is a link $(\gamma,\sigma)$ for some node $\sigma$. 
We perform the {\em scouting report} by computing the string $\gamma'$ that would be TP$[s+1]$ were there no link. If $\gamma' <_{L} \sigma$, remove the link $(\gamma,\sigma)$, let TP$[s+1]_{t+1} = \gamma' \upharpoonright (|\gamma| +1)$, and go to substage $t+2$. If $\gamma' \not<_{L} \sigma$, see whether $l_1(\gamma)[s+1] > p(\sigma), \delta_e(\gamma_e(p(\sigma)))[s+1]$ with $\gamma_e(p(\sigma))[s+1] > d_{\gamma}[s+1]$.


\emph{Subcase 1a.1}. No. Set TP$[s+1]_{t+1} = \tau \ \hat{\empty} \ f$ and go to substage $t+2$.

\emph{Subcase 1a.2}. Yes and for some node $\beta$ devoted to $\mathcal{H}_x$ with $m(\beta) = m$ and $\gamma \prec \beta \ \hat{\empty} \ \infty \preccurlyeq \sigma$, there is a marker $\lambda(x,m,q)$ with $p(\sigma)  \leqslant \lambda(x,m,q) \leqslant \delta_e(\gamma_e(l_1(\gamma)))[s+1]$. Our action is to set TP$[s+1]_{t+1} = \sigma$. We refer to this action as traveling the link. Go to substage $t+2$. 

\emph{Subcase 1a.3}. Otherwise, set TP$[s+1]_{t+1} = \sigma$ and go to substage $t+2$.

\emph{Subcase 1b}. There is no link from $\gamma$. Let TP$[s+1]_{t+1} = \gamma \ \hat{\empty} \ \infty$. 

\bigskip

\emph{Case 2}. $\gamma$ is a worker node devoted to $\mathcal{N}_{e,k}$.

\emph{Subcase 2a}. We were in subcase 1a.2 in the previous substage. Enumerate all markers 
as in the previous substage
into $A$. Let TP$[s+1]_{t+1} = \gamma \ \hat{\empty} \ g$.

\emph{Subcase 2b}. We were in subcase 1a.3 in the previous substage. Define $m_{\tau(\gamma)}$ to wager \$1 on $\Delta_e^A[s+1] \upharpoonright \gamma_e(l_1(\tau(\gamma)))[s+1]$ and bet neutrally on all other strings up to and including that length. For all $x,m$ such that there is $\beta \succ \gamma$ with $x(\beta) =x, m(\beta) = m$, let $c_{\tau(\gamma)}(x,m)[s+1] = \delta_e(\gamma_e(l_1(\tau(\gamma)))[s+1]$. Remove the link $(\tau(\gamma),\gamma)$. Let TP$[s+1]_{t+1} = \gamma \ \hat{\empty} \ d$. 

\emph{Subcase 2c}. We did not travel a link to arrive at $\gamma$. If $m_{\tau(\gamma)}(A_s \upharpoonright d_{\tau(\gamma)}[s+1]) \geqslant k$, let TP$[s+1]_{t+1} = \gamma \ \hat{\empty} \ d$, and go to substage $t+2$. If not, choose a fresh large follower $p(\gamma)$ for $\gamma$, place a link $(\tau(\gamma),\gamma)$, and go to stage $s+2$. 


\bigskip

\emph{Case 3}. $\gamma$ is a node devoted to $\mathcal{H}_x$ with $m(\gamma) = m$. Consider the immediate successors of $\gamma$ on the priority tree.

\emph{Case 3a}. The immediate successors of $\gamma$ are $\gamma \ \hat{\empty} \ \infty$ and $\gamma \ \hat{\empty} \ f$. See whether $l_2(\gamma)[s+1] > m_2(\gamma)[s+1]$.

\emph{Case 3a.1}.  Yes. For all $q \leqslant m_2(\gamma)[s+1]$, if $\lambda(x,m,q) > \max_{\tau \leqslant \gamma} c_{\tau}(x,m)[s+1]$ and $\Lambda^A(x,m,q) = 0$, enumerate $\lambda(x,m,q)$ into $A$ and define $\Lambda^A(x,m,q) = 1$ with use $-1$. Set TP$[s+1]_{t+1} = \gamma \ \hat{\empty} \ \infty$.

\emph{Case 3a.2}. No. Set TP$[s+1]_{t+1} = \gamma \ \hat{\empty} \ f$.

\emph{Case 3b}. The immediate successor of $\gamma$ is $\gamma \ \hat{\empty} \ \infty$. For all $q \leqslant m_2(\gamma)[s+1]$, if 
\[
\textrm{$\lambda(x,m,q) > \max_{\tau \leqslant \gamma} c_{\tau}(x,m)[s+1]$ and $\Lambda^A(x,m,q) = 0$}
\] 
enumerate $\lambda(x,m,q)$ into $A$ 
and define $\Lambda^A(x,m,q) = 1$ with use $-1$. Set TP$[s+1]_{t+1} = \beta \ \hat{\empty} \ \infty$.

\ \paragraph{{\bf Verification}}

We define TP, the {\em true path}, to be the leftmost path visited infinitely often. This clearly exists as the priority tree is finite-branching. We define GTP, the {\em genuine} true path, to be those $\alpha \prec$ TP such that there are infinitely many {\em genuine} $\alpha$-stages.

\begin{lem} 
For every $e \in \omega$ there is a node $\tau$ devoted to $\mathcal{N}_e$ on GTP.
\end{lem}
\begin{proof}
Let $\tau$ be the longest node devoted to $\mathcal{N}_e$ on TP, which exists by the finite injury along any path lemma. We claim that $\tau$ is on GTP. Suppose otherwise. Then it must be the case that there are $\tau'$ and $\sigma$ on GTP such that $\tau' \prec \tau \prec \sigma$ and the link $(\tau', \sigma)$ is there at almost all stages. This implies that $\sigma \ \hat{\empty} \ g$ is on GTP. On the priority tree, if $\tau_1 \prec \tau_2$ then $e(\tau_1) < e(\tau_2)$. Now as $\sigma$ links to its mother node $\tau'$ and $\tau' \prec \tau$ we must have $e(\tau') < e(\tau) = e$. But then by the construction of the priority tree, there is an $\mathcal{N}_e$ node below $\sigma \ \hat{\empty} \ g$, contradicting the hypothesis that $\tau$ is the longest such.
\end{proof}

\begin{lem}
For every $e \in \omega$, $\mathcal{N}_e$ is satisfied.
\end{lem}
\begin{proof}
Let $\tau$ be the longest node on GTP devoted to $\mathcal{N}_e$. First suppose that $\tau \ \hat{\empty} \ f \prec$ GTP. Then $\Gamma_e^{\Delta_e^A} \ne A$ and $\mathcal{N}_e$ is vacuously satisfied. If there is a permanent link $(\tau,\sigma)$ for some node $\sigma$, then we claim that $\Gamma_e^{\Delta_e^A} \ne A$. For contradiction suppose $\Gamma_e^{\Delta_e^A} = A$, and suppose the link was placed at stage $s_0$. Then there is a stage $s > s_0$ and uses $\delta_e(\gamma_e(l_1(\tau))) = a_1, \gamma_e(l_1(\tau)) = a_2$, such that 
$\Gamma_e^{\Delta_e^{A \upharpoonright a_1} \upharpoonright a_2}[s] \preccurlyeq A \upharpoonright a_1$. However if this were the case, at the next genuine $\tau$-stage greater than $s$ we will see that these computations have converged, play the $d$ outcome, and remove the link. This contradicts the fact that the link is permanent. This establishes the claim.

Now suppose there is no such permanent link. We will show that $m_{\tau}$ succeeds on $\Delta_{e}^A$. 
We must ensure that $m_{\tau}$'s capital does not decrease along $\Delta_e^A$. Fix $k \in \omega$, and let $\sigma' \succ \tau$ be the first node devoted to $\mathcal{N}_{e,k}$ that is visited. Suppose we visit $\sigma'$ first at stage $s_0$. At stage $s_0$ we assign $\sigma'$ a fresh large location marker $p(\sigma')$ and link back to $\tau$. Let $s_1$ be the stage at which we first define $m_{\tau}$ to win \$$k$ on $\Delta_e^{A}[s_0]$. As we acted in case (2b) of the construction, we did not play the $g$ outcome at stage $s_1$ and so there were no markers belonging to any nodes $\beta$ such that $\tau \prec \beta \ \hat{\empty} \ \infty \preccurlyeq \sigma'$ below the use of our computations. The computations are $\tau$-correct at stage $s_1$ and restraint is imposed on requirements of weaker priority than $\sigma'$. Therefore the only markers which can be enumerated below the use are those belonging to nodes $\beta$ such that $\tau \prec \beta \ \hat{\empty} \ f \preccurlyeq \sigma'$. As $\beta \ \hat{\empty} \ f$ was visited at stage $s_0$, there is at least one marker, namely $\lambda(x,m,s_0)$ where $x = x(\beta)$ and $m = m(\beta)$, that has not been enumerated into $A$ by stage $s_0$. The location marker $p(\sigma')$ was chosen to be large at the substage when $\sigma'$ was visited, and so is larger than $\lambda(x,m,s_0)$. If at some later stage $t$ we enumerate $\lambda(x,m,s_0)$ into $A$, then $A$ changes below $p(\sigma')$. If we are to have $\Gamma_e^{\Delta_e^A} = A$, then $\Delta_e^A[t] \upharpoonright \gamma_e(p(\sigma'))[s_0]$ is incomparable with $\Delta_e^A[s_0] \upharpoonright \gamma_e(p(\sigma'))[s_0]$. Therefore $m_{\tau}(\Delta_e^A[t] \upharpoonright d_{\tau}[t]) = k-1$, and so $m_{\tau}$ has not decreased in capital along $\Delta_e^A$. 
If we later arrive at $\beta$, we will play $\beta \ \hat{\empty} \ \infty$ and initialize $\sigma'$ as it is to the right of $\beta \ \hat{\empty} \ \infty$. If we visit $\sigma'$ again, a new location marker will be chosen, which must be larger than at least one marker of any node $\beta$ such that $\tau \prec \beta \ \hat{\empty} \ f \preccurlyeq \sigma'$.

Let $s_0$ be the least genuine $\tau$-stage. As $\tau$ is genuinely visited at stage $s_0$ there can be no link $(\tau',\sigma')$ at stage $s_0$ with $\tau' \prec \tau \prec \sigma'$. Fix $k \in \omega$ and suppose for contradiction that $\sigma \prec$ TP devoted to $\mathcal{N}_{e,k}$ is never genuinely visited. Then there is some permanent link $(\tau'', \sigma'')$ with $\tau \prec \tau'' \prec \sigma \prec \sigma''$. Suppose the link $(\tau'', \sigma'')$ was placed at stage $s_1$. Then as $\sigma''$ is genuinely visited at stage $s_1$ and $\sigma \prec \sigma''$, $\sigma$ is genuinely visited at stage $s_1$. Contradiction. At stage $s_1$ we define the location marker $p(\sigma)$ and create a link $(\tau,\sigma)$. Let $s_2$ be the stage at which we travel the link to $\sigma$ and define $m_{\tau}$ to win $\$k$ on $\Delta_e^A[s_2]$. If we visit a node that is below $\tau$ but to the left of $\sigma$ then, as in the previous paragraph, $m_{\tau}$ will then have capital $k-1$. However $\sigma \prec$ TP, and so this will occur only finitely many times. Let $s_3$ be the greatest stage at which $\sigma$ is initialized. We remove any link over $\sigma$ as part of the initalization. At the next $\sigma$-stage after $s_3$ we will genuinely visit $\sigma$ and place a link if we see that $m_{\tau}(\Delta_e^A[s_3] \upharpoonright d_{\tau}[s_3]) < k$. Let $s'$ be the stage at which we travel the link to $\sigma$ and define more of $m_{\tau}$. As we never visit a node to the left of $\sigma$, no marker belonging to a node $\beta$ such that $\tau \prec \beta \ \hat{\empty} \ f \preccurlyeq \sigma$ will be enumerated after stage $s'$. As in the previous paragraph, $A$ cannot change below the use $\delta_e(\gamma_e(l(\tau))[s']$ and $m_{\tau}(\Delta_e^A) \geqslant k$. 
\end{proof}

\begin{lem}
For every $x \in \omega$, $\mathcal{H}_x$ is satisfied.
\end{lem}
\begin{proof}
Suppose $ x \not\in$ Cof. Let $\beta \prec$ TP be the node devoted to $\mathcal{H}_x$ with $m(\beta)$ least such that $\beta_{x,m}$ is not permanently linked over by $\tau$'s of stronger priority for all $m \geqslant m(\beta)$. Let $m(\beta) = m_0$. We show that $\lim_t \Lambda^A(x,m,t) = 0$ for all $m \geqslant m_0$. 
We will have $\Lambda^A(x,m,t) = 0$ unless the marker $\lambda(x,m,t)$ is enumerated into $A$. Thus we must show that we eventually stop enumerating the markers $\lambda(x,m,t)$ into $A$. As $[m_0, \infty) \not\subseteq W_x$, the $\Pi_2$ test which measures whether $[m,\infty) \subseteq W_x$ will eventually always say ``no''. So eventually the $\beta_{x,m}$ nodes will stop putting their markers into $A$ and redefining $\Lambda^A(x,m,t)$. Therefore the only way we will enumerate the markers is if there is a link $(\tau,\sigma)$ over $\beta_{x,m}$ with $\tau \prec \beta_{x,m} \ \hat{\empty} \ \infty \preccurlyeq \sigma$ and $\tau$ is accessible. We must show that if there is such a link then $\tau$ is accessible at only finitely many stages. 

For $\tau$ with $e(\tau) > x$, we will restart $\tau$ below $\beta_{x,m} \ \hat{\empty} \ \infty$. If $\beta \ \hat{\empty} \ \infty \preccurlyeq \tau \prec \beta_{x,m_1} \preccurlyeq \sigma$ and a link $(\tau,\sigma)$ is placed over $\beta_{x,m_1}$ for some $m_1 > m$, then as $\beta \ \hat{\empty} \ \infty$ can be visited at most finitely many times, only finitely many markers $\lambda(x,m_1,t)$ will be enumerated.

Now suppose $\tau$ is above $\beta$ and we place a link $(\tau,\sigma)$ over $\beta$. Subcase (2a) of the construction will enumerate markers $\lambda(x,m,t)$ only if $\beta_{x,m} \ \hat{\empty} \ \infty \preccurlyeq \sigma$. As $\sigma$ must be below $\beta \ \hat{\empty} \ f$, the markers $\lambda(x,m,t)$ will not be enumerated, and we will have $\lim_t \Lambda^A(x,m,t) = 0$.

Finally, if $\beta'_{x,m}$ is another 
node on another path of the priority tree which is visited infinitely often, we must ensure that it does not enumerate all of the markers $\lambda(x,m,t)$. The $\Pi_2$ test performed at $\beta'_{x,m}$ is the same test which is performed at $\beta_{x,m}$ and so will eventually always say ``no''. Therefore the marker $\lambda(x,m,t)$ will only be enumerated if there is a link $(\tau,\sigma)$ over $\beta'_{x,m}$ with $\beta'_{x,m} \ \hat{\empty} \ \infty \preccurlyeq \sigma$. As $\beta'_{x,m}$ is to the right of TP it will be initialized infinitely many times. Any link over $\beta'_{x,m}$ will be removed as part of the initialization, and so $\beta'_{x,m}$ is not permanently linked over. The marker $\lambda(x,m,t)$ is enumerated capriciously only if $\beta'_{x,m} \ \hat{\empty} \ \infty \preccurlyeq \sigma$. For the marker to be enumerated infinitely many times we must visit $\sigma$ below $\beta'_{x,m} \ \hat{\empty} \ \infty$ and place a link back to $\tau$. However if $\beta'_{x,m} \ \hat{\empty} \ \infty$ is visited infinitely many times, this contradicts $x \not\in$ Cof. Therefore $\beta'_{x,m}$ cannot enumerate infinitely many of its markers.

Now suppose that $x \in$ Cof and $[m_0,\infty) \subseteq W_x$. Let $\beta \prec$ TP be the node devoted to $\mathcal{H}_x$ with $m(\beta)$ least such that $\beta_{x,m}$ is not permanently linked over by $\tau$'s of stronger priority for all $m \geqslant m(\beta)$. Let $m' = \max \{ m(\beta), m_0 \}$. We show that $\lim_t \Lambda^A(x,m,t) = 1$ for all $m \geqslant m'$.  

We first show that $\lim_s c_{\tau}(x,m)[s]$ exists for all $\tau \leqslant \beta$ and so all markers $\lambda(x,m,q) > \max_{\tau \leqslant \beta} \lim_s c_{\tau}(x,m)[s]$ may be enumerated into $A$. The value of $c_{\tau}(x,m)$ can be increased only when the mother node $\tau$ 
links to some worker $\sigma$ with $\sigma \prec \beta'_{x,m}$ and subsequently defines more of the martingale $m_{\tau}$. The priority tree is finite-branching, and so there are only finitely many nodes $\beta_{x,m}$. Consider
 \[
 \langle e', k' \rangle = \max \cup_{\beta_{x,m}} \{ \langle e,k \rangle\ |\  \sigma \prec \beta_{x,m} \mbox{ has mother $\tau$ and is devoted to } \mathcal{N}_{e,k} \}.
 \] 
 Let $\sigma_0$ be the node on TP devoted to $\mathcal{N}_{e',k'}$, and suppose that no node to the left of $\sigma_0$ is visited after stage $s_0$. When we genuinely visit any node $\sigma \preccurlyeq \sigma_0$ with mother node $\tau$, the construction will check to see whether $m_{\tau}(\Delta_e^A[s] \upharpoonright d_{\tau}[s]) \geqslant k(\sigma)$. If so, we will play the $d$ outcome. As $\sigma \prec$ TP, the martingale will never decrease in capital and $c_{\tau}(x,m)$ cannot be increased again after travelling a link to some $\sigma'$ devoted to $\mathcal{N}_{e,k(\sigma)}$. If $m_{\tau}(\Delta_e^A \upharpoonright d_{\tau}[s]) < k(\sigma)$, then we will create a link from $\sigma$ back to $\tau$ at stage $s > s_0$. If the link is permanent then $c_{\tau}(x,m)$ will never be increased. If we later define more of $m_{\tau}$ we will then increase $c_{\tau}(x,m)$. Again, as $\sigma \prec$ TP, the martingale will never decrease in capital and $c_{\tau}(x,m)$ cannot be increased again after travelling a link to some $\sigma'$ devoted to $\mathcal{N}_{e,k(\sigma)}$.
There are only finitely many worker nodes $\sigma \preccurlyeq \sigma_0$ with mother node $\tau$. Therefore $\lim_s c_{\tau}(x,m)[s]$ exists for all $\tau \leqslant \beta$.

We will restart all $\tau$ nodes with $e(\tau) > x$ below $\beta_{x,m'} \ \hat{\empty} \ \infty$. As $[m', \infty) \subseteq W_x$, the $\Pi_2$ test which measures whether $[m,\infty) \subseteq W_x$ will say ``yes'' infinitely many times for all $m \geqslant m'$. If a link $(\tau,\sigma)$ is placed over $\beta_{x,m'}$, and so $\beta_{x,m'} \ \hat{\empty} \ f \preccurlyeq \sigma$, we will perform a scouting report when we arrive at $\tau$. Suppose the link $(\tau,\sigma)$ is placed at stage $s_0$ and that $s_1$ is the least stage greater than $s_0$ at which the $\Pi_2$ test says ``yes''. At the least $\tau$-stage after $s_1$, the scouting report will be successful, and $\beta_{x,m'} \ \hat{\empty} \ \infty$ will be accessible. We then will remove the link $(\tau,\sigma)$ and enumerate the marker $\lambda(x,m',t)$. In this way all of $\beta_{x,m'}$'s markers will eventually be enumerated.

Now suppose that $\beta_{x,m'} \ \hat{\empty} \ \infty \preccurlyeq \tau \prec \beta_{x,m_1} \prec \sigma$. The node $\beta_{x,m_1}$ has only the $\infty$ outcome, as it is below $\beta_{x,m'} \ \hat{\empty} \ \infty$ with $m' < m_1$. If $\beta_{x,m_1}$ is on GTP then it will enumerate its markers whenever it is accessible and define $\Lambda$ such that $\lim_t \Lambda^A(x,m_1,t) = 1$. If $\beta_{x,m_1}$ is not on GTP then the link $(\tau,\sigma)$ must be permanent. As in the previous paragraph, $\beta_{x,m'} \ \hat{\empty} \ \infty$ will be accessible at infinitely many stages. When we arrive at $\tau$ and see the link, we will enumerate markers of the form $\lambda(x,m_1,q)$. Therefore all of $\beta_{x,m_1}$'s markers will eventually be enumerated and $\lim_t \Lambda^A(x,m_1,t) = 1$.
\end{proof}
\noindent
This concludes the verification of the construction, and the proof of Theorem
\ref{th:high2cenoivr}.

\end{document}